\documentclass[11pt,draftclsnofoot,onecolumn]{IEEEtran}
\pdfoutput=1
\ifCLASSINFOpdf
\else
\fi

\usepackage[table]{xcolor}
\usepackage{cite}
\usepackage{amsmath,amssymb,amsfonts}
\usepackage{amsthm}
\usepackage{algorithmic}
\usepackage{graphicx}
\usepackage{textcomp}
\usepackage{xcolor}
\usepackage{caption}
\usepackage{multirow}
\usepackage{array}
\usepackage{tabularx}
\usepackage{marginnote}
\newtheorem{theorem}{Theorem}
\newtheorem{definition}{Definition}
\newtheorem{remark}{Remark}
\newtheorem{prop}{Proposition}

\newtheorem{assum}{Assumption}

\usepackage{epstopdf}
\usepackage{epsfig}
\usepackage{subfigure} %

\usepackage{color}
\definecolor{mygray}{gray}{.9}

\begin{document}
%
\title{Learning Channel Capacity with Neural Mutual Information Estimator Based on Message Importance Measure}
%
%
%

\author{\IEEEauthorblockN{
Zhefan Li, Rui She, Pingyi Fan, \IEEEmembership{Senior Member IEEE}, Chenghui Peng, and Khaled B. Letaief, \IEEEmembership{Fellow IEEE}\\}
\thanks{This work was supported by the National Key R \& D Program of China under Grant 2021YFA1000500(4). K. B. Letaief's work was partially supported by the Hong Kong Research Grants Council Area of Excellence Scheme under Grant No.AoE/E-601/22R.}%
\thanks{Z. Li and P. Fan are with the Department of Electronic Engineering, Tsinghua University, Beijing 100084, China and Beijing National Research Center for Information Science and Technology (BNRist), Beijing 100084, China (e-mail: lzf20@mails.tsinghua.edu.cn; fpy@tsinghua.edu.cn).}
\thanks{R. She is with the School of Electrical and Electronic Engineering, Nanyang Technological University (NTU), 637553, Singapore (e-mail: rui.she@ntu.edu.sg).}
\thanks{C. Peng is with the Wireless Technology Laboratory, Huawei Technologies, Shenzhen 518129, China (e-mail: pengchenghui@huawei.com).}
\thanks{K. B. Letaief is with the Department of Electronic Engineering, The Hong Kong University of Science and Technology, Hong Kong (e-mail: eekhaled@ust.hk).}
}

\maketitle

\begin{abstract}
Channel capacity estimation plays a crucial role in beyond 5G intelligent communications. Despite its significance, this task is challenging for a majority of channels, especially for the complex channels not modeled as the well-known typical ones. Recently, neural networks have been used in mutual information estimation and optimization. They are particularly considered as efficient tools for learning channel capacity. In this paper, we propose a cooperative framework to simultaneously estimate channel capacity and design the optimal codebook. First, we will leverage MIM-based GAN, a novel form of generative adversarial network (GAN) using message importance measure (MIM) as the information distance, into mutual information estimation, and develop a novel method, named MIM-based mutual information estimator (MMIE). Then, we design a generalized cooperative framework for channel capacity learning, in which a generator is regarded as an encoder producing the channel input, while a discriminator is the mutual information estimator that assesses the performance of the generator. Through the adversarial training, the generator automatically learns the optimal codebook and the discriminator estimates the channel capacity. Numerical experiments will demonstrate that compared with several conventional estimators, the MMIE achieves state-of-the-art performance in terms of accuracy and stability.
\end{abstract}



%
\IEEEpeerreviewmaketitle

\section{Introduction}
%
%
%
%
\IEEEPARstart{C}{hannel} capacity is a pivotal quantity in information theory\cite{r1,r2,r3}, communication industry\cite{r4,r5} and computer science\cite{r6,r7}, as it represents the maximum information rate that a communication system can reliably transmit. Shannon first defined it in 1948 and gave the Shannon-Hartley Theorem to calculate the channel capacity for the additive white Gaussian noise (AWGN) channel\cite{r8}. However, for general channels, the channel capacity does not usually have a closed-form solution and it is challenging to compute it \cite{r9,nguyen2010estimating,wunder2021reverse}. Consequently, it is critical to be able to estimate the channel capacity for complex communication scenarios in practice.
In this regard, there are numerous proposed algorithms to overcome this challenge.

Most of the conventional channel capacity estimation algorithms are designed based on iterative or numerical methods. These include the linearly constrained optimization approach \cite{r10,r11}, Arimoto-Blahut algorithm\cite{r12,r13}, and simulation-based numerical algorithms\cite{r14,r25}.
Despite their success in the discrete cases, few of them perform well for the continuous input channels. Moreover, since mutual information can be analytically calculated by its definition in the discrete case, conventional algorithms do not contain a specific sample-based mutual information estimator, which is essential for handling general continous channels. 

Recently, with the development of deep learning, there has been great progress in mutual information estimation \cite{r9,r15,r21,r28,r26,r27}. For example, \cite{r15} combines variational approaches with neural networks, and provides the first neural network based estimator of mutual information. Various
researchers have considered extension and requirements to this
original estimator, which is discussed in detail in Section \ref{sect.NCCE}. These estimators are particularly effective
at estimating mutual information from sample data, which motivates researchers to estimate channel capacity by maximizing the output of these estimators\cite{r9,r26,r27,r29,r30}.

The essence of using a mutual information estimator to estimate the channel capacity is revealed by the relationship between the mutual information and the channel capacity. When two random variables $X$ and $Y$ are the input and output for an arbitrary memoryless stochastic channel, the channel capacity is given by
\begin{equation}
	C = \mathop{\max}\limits_{p_X(\mathbf{x})} I(X,Y),
\end{equation}
where $I(X,Y)$ is the mutual information \cite{r8} defined as
\begin{equation}
    \begin{split}
	I(X,Y)
    & = D_{\text{KL}}(p_{XY}(\mathbf{x},\mathbf{y})\|p_X(\mathbf{x})p_Y(\mathbf{y})) \\
    & =\mathbb{E}_{(\mathbf{x},\mathbf{y})\sim p_{XY}(\mathbf{x},\mathbf{y})}\left[\ln\frac{p_{XY}(\mathbf{x},\mathbf{y})}{p_X(\mathbf{x})p_Y(\mathbf{y})}\right].
    \end{split}
\end{equation}
Here $D_{\text{KL}}(\cdot)$ is the Kullback-Leibler (K-L) divergence operator, $p_{XY}$ denotes the joint distribution. Also $p_X$ and $p_Y$ are the marginal distributions for $X$ and $Y$, respectively.

In this case, channel capacity estimation can be regarded as two fundamental tasks\cite{r9}, namely, (a) estimating the mutual information between the channel input and output, as well as, (b) maximizing the mutual information with respect to the channel input distribution. When the mutual information estimator is differential, the latter is easily achieved using gradient descent methods. As for the former, traditional methods can be used, such as binning\cite{r17}, non-parametric kernel estimation\cite{r18,r19}, and approximation of the Gaussian distribution\cite{r20}. However, these methods are not scalable and do not perform with large sample sizes and dimensions, especially for high dimensional data.
To this end, deep learning methods with high dimensional encoders were considered as efficient estimators \cite{r9,r15,r21}.

The comparisons of different channel capacity estimation approaches are summarized in Table \ref{table.3}. Conventional channel capacity estimation algorithms do not contain specific mutual information estimators and fail to deal with continuous channels. By contrast, the two-step optimization method divides the capacity estimation into two steps: (a) mutual information estimation, as well as, (b) mutual information optimization. With respect to step (a), neural mutual information estimators outperform traditional ones. Besides, the performance of the capacity estimators extremely relies on the choice of the mutual information estimator \cite{r31}. Hence, it is promising to study novel neural mutual information estimators. One specific algorithm is achieved in \cite{r9}. The authors suggest that the mutual information estimator can be designed according to the objective function of the generative adversarial networks (GAN) and provide a cooperative framework for capacity learning based on GAN.

Inspired by that, we make use of MIM-based GAN \cite{r16} which provides a more stable generator and a more accurate discriminator, to further improve the performance of deep learning-based mutual information estimator.
In particular, we first propose the MIM-based mutual information estimators (MMIE), including non-parametric MMIE and parametric MMIE. 
Furthermore, we present a corresponding framework for both channel capacity estimation and codebook design.
In brief, our contributions lie in a novel learning-based mutual information estimator and a framework to estimate channel capacity and generate codebooks, simultaneously.

\begin{table}[htb]
\centering
\caption{\bf Comparison of Different Channel Capacity Estimation Approaches}\label{table.3}
\newcommand{\tabincell}[2]{\begin{tabular}{@{}#1@{}}#2\end{tabular}}
\begin{tabular}{|m{2.7cm}|m{1.7cm}|m{2.86cm}|m{3.4cm}|m{3.2cm}|}
\hline
\rowcolor{mygray}
\textbf{Category} & \textbf{Mutual information estimator} & \textbf{Approach} & \textbf{Advantage} & \textbf{Disadvantage} \\
\hline
\tabincell{l}{\tabincell{l}{Direct \\ optimization}}
& \tabincell{l}{analytically \\ calculated}
& \tabincell{l}{$\bullet$ Linearly constrained \\ optimization;\\ $\bullet$ Arimoto-Blahut\\ algorithm; \\ $\bullet$  Simulation-based\\ method.}
& \tabincell{l}{$\bullet$ Successful in the\\ discrete cases.}
& \tabincell{l}{$\bullet$ Can not deal with\\ continuous channels.} \\
\hline
\multirow{2}{*}{\tabincell{l}{Two-step \\ optimazation}}
		 & \tabincell{l}{Traditional \\ methods} & \tabincell{l}{$\bullet$ Binning method;\\ $\bullet$ Non-parametric\\ kernel estimation; \\ $\bullet$  Approximation of\\ Gaussian distributions.} & \tabincell{l}{$\bullet$ Explainable in theory;\\ $\bullet$ Simple for achievements;\\ $\bullet$ Based on sampling data;\\ $\bullet$ Available for discrete\\ and continuous cases.} & \tabincell{l}{$\bullet$ Not scalable;\\  $\bullet$ Not adaptive to \\ large sample sizes and \\ dimensions.}\\
		\cline{2-5} & \tabincell{l}{Neural \\ networks} & \tabincell{l}{$\bullet$ GAN-based\cite{r9} \\ $\bullet$ Combined with \\autoencoder \cite{r26,r27} \\ $\bullet$ Consider feedforward \\and feedback \\channels\cite{r30}\\ $\bullet$ Consider multi \\access channels\cite{r29}}   & \tabincell{l}{$\bullet$ For general channels; \\ $\bullet$ Empirical effectiveness;\\ $\bullet$ Efficient for high \\ dimensional cases.} & \tabincell{l}{$\bullet$ Performance relies on \\ the neural mutual \\ information estimator.}
\\
\hline
\end{tabular}
\end{table}

The rest of this paper is organized as follows. In Section II, we review mutual information estimators based on neural networks. In Section III, we first give a brief introduction to MIM-based GAN and present the new mutual information estimator named MMIE. Then, we discuss how the estimator helps to estimate the  channel capacity in Section IV. Section V presents sample simulation results. Finally, we conclude the paper in Section VI.


\section{Related Works}
\subsection{Neural Channel Capacity Estimation}\label{sect.NCCE}
The research of neural channel capacity estimation algorithms can be generally classified into two categories. One focuses on designing effective neural mutual information estimators with low bias and variance\cite{r22,r15,r21,r9,r30}. Their performance is usually measured by estimating the capacity of a memoryless point-to-point channel since solving this problem needs to directly maximize the output of the mutual information estimator\cite{r31}. The other tends to extend the range of the applications based on the current estimators. For example, \cite{r30} presents DINE for evaluating 
the feedforward capacity and feedback capacity of continuous channels
by leveraging the mutual information estimator in \cite{r32}. \cite{r26} combines the estimator in \cite{r15} with a cross-entropy-based autoencoder to provide a capacity-approaching codebook. Besides, \cite {r29} considers the capacity region of the multiple-access channels. Different from most studies that estimate the lower bound of the capacity, \cite{r33} considers the problem of estimating an upper capacity bound by exploiting the dual representation of
channel capacity based on a variation of the estimator in \cite{r15}. Our work belongs to the first category, i.e., to design novel mutual information estimators for accurate capacity evaluation. In the following section, we shall introduce two families of neural mutual information estimators.
	
\subsection{Variational Mutual Information Estimator}

In order to overcome the disadvantages of the conventional mutual information estimators, recent works attempted to combine the variational estimation approaches with deep learning, i.e., utilizing neural networks to maximize the variational lower bound of mutual information \cite{r21}. Though these algorithms have confirmed their feasibility and effectiveness in practice, there exist challenges to explain their theoretical effectiveness \cite{r9,r21}. Besides, these estimators usually suffer from either high bias or high variance\cite{r21}.
	
Generally speaking, the variational mutual information estimators can be summarized into generative and discriminative approaches. On one hand, the former, namely generative approaches, estimate the joint and marginal distributions separately. For example, the Barber-Agakov (BA) estimator
is introduced in \cite{r22} to learn a conditional generative model, which is given by
\begin{equation}
	\begin{split}
		& I_{\text{BA}}(X,Y) \\
		&=\mathbb{E}_{(\mathbf{x},\mathbf{y})\sim p_{XY}(\mathbf{x},\mathbf{y})}[\ln(q_{X,\phi}(\mathbf{x}|\mathbf{y}))-\ln(p_X(\mathbf{x}))] \\
        & \le I(X;Y),
	\end{split}	
\end{equation}
in which $q_{X,\phi}(\mathbf{x})$ is an arbitrary conditional distribution, parameterized by $\phi$.
On the other hand, the latter, namely discriminative approaches, directly estimate the density ratio $R=\frac{p_{XY}(\mathbf{x},\mathbf{y})}{p_X(\mathbf{x})p_Y(\mathbf{y})}$. One family relies on the variational lower bounds for K-L divergence\cite{r21}.
For instance, the Mutual Information Neural Estimator (MINE) \cite{r15} is designed based on an approximation of Donsker-Varadhan representation of the K-L divergence,
which is given by
\begin{equation}
	\begin{split}
		I(X;Y)\ge&I_{\text{MINE}}(X,Y) \\
			  =&\mathop{\sup}\limits_{\theta \in \Theta} \{ \mathbb{E}_{(\mathbf{x},\mathbf{y})\sim p_{XY}(\mathbf{x},\mathbf{y})}[T_{\theta}(\mathbf{x},\mathbf{y})] \\
			  &-\ln(\mathbb{E}_{\mathbf{x},\mathbf{y}\sim p_X(\mathbf{x})p_Y(\mathbf{y})}[\exp{(T_{\theta}(\mathbf{x},\mathbf{y}))}]) \},
	\end{split}
\end{equation}
where $T_{\theta}$ is a neural network parameterized by $\theta \in \Theta$. However, replacing the expectation operation with the sample average operation leads to bias for MINE. Fortunately, the bias can be reduced by using an exponential moving average method to estimate the partition
function $\mathbb{E}_{p_Xp_Y}[\exp{(T_{\theta})}]$ \cite{r15}.

Since the K-L divergence can be considered as a special case of $f$-divergence, another lower bound based on $f$-divergence representation is used to estimate the mutual information \cite{r15}. It is given by
\begin{equation}\label{eq.5}
	\begin{split}
		I(X,Y)\ge&I_{\text{NWJ}}(X,Y) \\
			   =&\mathop{\sup}\limits_{\theta \in \Theta} \{ \mathbb{E}_{(\mathbf{x},\mathbf{y})\sim p_{XY}(\mathbf{x},\mathbf{y})}[T_{\theta}(\mathbf{x},\mathbf{y})] \\
			   &-\mathbb{E}_{\mathbf{x},\mathbf{y}\sim p_X(\mathbf{x})p_Y(\mathbf{y})}[\exp{(T_{\theta}(\mathbf{x},\mathbf{y})-1)}] \},
	\end{split}
\end{equation}
in which the corresponding estimation method is named the
 Nguyen-Wainwright-Jordan (NWJ) estimator and will be denoted by $I_{\text{NWJ}}(X,Y)$.
Unlike MINE, the NWJ estimator is unbiased and provides a looser bound than MINE. That is,  $I_{\text{NWJ}}(X,Y)\le I_{\text{MINE}}(X,Y)$.

Moreover, the theoretical analysis reveals that the variance of the estimators including MINE and NWJ increases exponentially with the ground truth of mutual information\cite{r21}. Inspired by this, a new estimator, named 
 the Smoothed Mutual Information “Lower-bound” Estimator (SMILE) is proposed to focus on variance reduction \cite{r21}, which is given by
\begin{equation}
	\begin{split}
		I(X,Y)\ge &I_{\text{SMILE}}(X,Y) \\
			  =&\mathop{\sup}\limits_{\theta \in \Theta}\{ \mathbb{E}_{(\mathbf{x},\mathbf{y})\sim p_{XY}(\mathbf{x},\mathbf{y})}[T_{\theta}(\mathbf{x},\mathbf{y})]\\
			  &-\ln(\mathbb{E}_{\mathbf{x},\mathbf{y}\sim p_X(\mathbf{x})p_Y(\mathbf{y})}[\mathrm{clip}({\rm e}^{T_{\theta}(\mathbf{x},\mathbf{y})},{\rm e}^{-\tau},{\rm e}^{\tau})]) \},
	\end{split}
\end{equation}
where the clip function is defined as
\begin{equation}
	\mathrm{clip}(v,l,u)=\max(\min(u,v),l),
\end{equation}
and it clips the log-density estimation function $T_{\theta}(\mathbf{x},\mathbf{y})$ between $-\tau$ and $\tau$. By adjusting the parameter $\tau$, SMILE realizes the trade-off between the bias and variance. Besides, it also converges to MINE as $\tau \to \infty$.

\subsection{GAN-Based Mutual Information Estimator}
GAN is proposed as an efficient method to capture a target distribution, especially for complex and high dimensional cases\cite{r23}. The framework of GAN consists of two networks that compete with each other. Specifically, the generator network $G$ produces samples $\mathbf{x} \sim p_{\text{fake}}(\mathbf{x})$ to confuse the discriminator network while the discriminator network $D$ strives to distinguish those generative samples from real data $\mathbf{x} \sim p_{\text{real}}(\mathbf{x})$. In \cite{r23}, it is noticed that the optimal solution of the discriminator is  directly related to the density ratio $\frac{p_{\text{fake}}}{p_{\text{real}}}$.
Inspired by this, \cite{r9} proposes two discriminative mutual information estimators
by setting $p_{\text{fake}}\equiv p_Xp_Y$ and $p_{\text{real}}\equiv p_{XY}$. Then, we give a brief introduction to them. One of them named iDIME is designed as follows.

First, a discriminator objective function $J(D)$ is defined as
\begin{equation}
	\begin{split}
		J(D)=&\mathbb{E}_{\mathbf{x},\mathbf{y}\sim p_X(\mathbf{x})p_Y(\mathbf{y})}[\ln(D(\mathbf{x},\mathbf{y}))]\\
		&+\mathbb{E}_{\mathbf{x},\mathbf{y}\sim p_{XY}(\mathbf{x},\mathbf{y})}[\ln(1-D(\mathbf{x},\mathbf{y})))],
	\end{split}
\end{equation}
where $D$ denotes the discriminator. Then, the corresponding optimal discriminator is obtained as
\begin{equation}
D^*(\mathbf{x},\mathbf{y})
=\mathop{\arg\max}\limits_{D} J(D)
=\frac{p_X(\mathbf{x})p_Y(\mathbf{y})}{p_{XY}(\mathbf{x},\mathbf{y})+p_X(\mathbf{x})p_Y(\mathbf{y})}.
\end{equation}
Furthermore, the indirect discriminative mutual information estimator (iDIME) is given by
\begin{equation}
	I(X,Y)=I_{\text{iDIME}}(X,Y)=\mathbb{E}_{\mathbf{x},\mathbf{y}\sim p_{XY}(\mathbf{x},\mathbf{y})}\left[ln\frac{1-D^*(\mathbf{x},\mathbf{y})}{D^*(\mathbf{x},\mathbf{y})}\right].
\end{equation}

Different from the conventional variational approaches, the objective function $J(D)$ is not directly related to the mutual information. Further, when $p_X(\mathbf{x})p_Y(\mathbf{y})\ll p_{XY}(\mathbf{x},\mathbf{y})$ or $p_X(\mathbf{x})p_Y(\mathbf{y})\gg p_{XY}(\mathbf{x},\mathbf{y})$, $D(\mathbf{x},\mathbf{y})$ may saturate to $1$ or degenerate to $0$, causing unstable estimation results. Hence, another estimator named direct discriminative mutual information estimator (dDIME)
is studied, whose objective function is defined as
\begin{equation}
	\begin{split}
		J_{\alpha}(D)=&\alpha\mathbb{E}_{\mathbf{x},\mathbf{y}\sim p_{XY}(\mathbf{x},\mathbf{y})}[\ln(D(\mathbf{x},\mathbf{y}))]\\
		&+\mathbb{E}_{\mathbf{x},\mathbf{y}\sim p_X(\mathbf{x})p_Y(\mathbf{y})}[-D(\mathbf{x},\mathbf{y})],
	\end{split}
\end{equation}
where $\alpha$ is an adjustable parameter and $\alpha>0$.
Then, similar to iDIME, it is readily seen that
\begin{equation}
D^*(\mathbf{x},\mathbf{y})
=\mathop{\arg\max}\limits_{D} J_{\alpha}(D)
=\alpha\frac{p_{XY}(\mathbf{x},\mathbf{y})}{p_X(\mathbf{x})p_Y(\mathbf{y})},
\end{equation}
and
\begin{equation}
	\begin{split}
		I(X,Y)&=I_{\text{dDIME}}(X,Y)\\
		&=\mathbb{E}_{\mathbf{x},\mathbf{y}\sim p_{XY}(\mathbf{x},\mathbf{y})}\left[\ln\left(\frac{D^*(\mathbf{x},\mathbf{y})}{\alpha}\right)\right] \\
        & = \frac{J_{\alpha}(D^*)}{\alpha}+1-\ln(\alpha).
	\end{split}	
\end{equation}
Since $J_{\alpha}(D^*)$ is the maximum value of $J_{\alpha}(D)$, dDIME gives a lower bound of the mutual information such as the variational approaches. It is given by
\begin{equation}
	\begin{split}
		I(X,Y)=&I_{\text{dDIME}}(X,Y) \\
			  \ge&\tilde{I}_{\text{dDIME}}(X,Y) \\
			  =&\frac{J_{\alpha}(D)}{\alpha}+1-\ln(\alpha).
	\end{split}
\end{equation}
from which it is easy to see that the parameter $\alpha$ in dDIME adjusts the range of the output of the discriminator. However, no explicit guideline for selecting $\alpha$ is provided in \cite{r9}.
Besides, the simulation in \cite{r9} shows that a high $\alpha$ may lead to numerical issues at high values of signal-to-noise ratio (SNR).
	
By resorting to dDIME, a cooperative framework for the codebook design is presented to approach the channel capacity \cite{r9}. Similar to GAN, the framework contains a pair of networks, that is, a generator $G$ providing the channel input $\mathbf{x}$, as well as a discriminator $D$ estimating the mutual information for the given channel input.
In the training process, the generator cooperates with the discriminator to maximize $\tilde{I}_{\text{dDIME}}(X,Y)$.

\begin{table}[htb]
	\centering
	\caption{\bf Comparision of Different Neural Mutual Information Estimators}\label{table.4}
	\begin{tabular}{|p{2.7cm}|p{1.7cm}|p{2.8cm}|p{3.2cm}|}
		\hline
        \rowcolor{mygray}
        \textbf{Category} & \textbf{Mutual information estimator} & \textbf{innovation} & \textbf{limitation} \\
		\hline
		\multirow{4}{2.7cm}{Variational estimation based}
		&BA & learn a conditional generative model & $H(x)$ needs to be known\\
		\cline{2-4} & MINE & first propose the neural estimator& biased; high variance\\
		\cline{2-4} &NWJ&unbiased form of MINE&high variance; less tight than MINE \\
		\cline{2-4} &SMILE&introduce a clip function to control the trade between bias and variance &no guideline for choosing the parameter of the clip function\\
		\hline
		\multirow{4}{*}{GAN based}
		&iDIME & based on the original GAN  & neither lower bound nor upper bound; unstable when SNR is high\\
		\cline{2-4} & dDIME & based on a new GAN proposed in \cite{r9}; lower bound of the mutual information & no guideline for choosing the parameter $\alpha$; unstable when SNR is high\\
		\cline{2-4} &MMIE\par(ours)&based on the MIM-based GAN; more stable& performance decline when SNR is high\\
		\cline{2-4} &$\alpha$-MMIE\par(ours) & parametric form of MMIE & acquirement for the parameter $\alpha$ \\
		\hline
	\end{tabular}
\end{table}
In summary, the mutual information estimators introduced above and our work are compared in Table \ref{table.4}.

\section{MIM-based Mutual information Estimator}
In this Section, we formulate a framework using MIM-based GAN to achieve a mutual information estimator and analyze its properties.

\subsection{Overview of MIM-based GAN}
The numerical experiments in \cite{r9} reveal that under the same network architecture and training settings, iDIME and dDIME perform differently in bias and variance.
Specifically, iDIME has more bias than dDIME in the case of  low SNR while iDIME has less bias in the case of high SNR. Moreover, for high dimensional data, iDIME has much less bias than dDIME in the case of high SNR with the expence of slightly high variance. Besides, the parameter $\alpha$ in the objective function of dDIME has an impact on the accuracy and stability of the estimator. These experimental results indicate that the performance of these estimators is related to the objective functions of the corresponding GAN they are based on.
Hence, it is promising to introduce other forms of GAN with different objective functions into estimator design. 

Recently, a novel form of GAN, named MIM-based GAN is proposed\cite{r16}. It makes use of a new information metric, message importance measure (MIM)\cite{r24}, to measure the information distance between the real and generative distributions, from which several improvements in network training and rare data generation are obtained.  
Specifically, the definition of MIM presented in \cite{r24} is given as follows.
\begin{definition}\label{def.MIM}
	For a discrete probability distribution $\{p_1,p_2,\cdots,p_N\}$, and a given importance coefficient $\bar{\omega}$, the message importance measure, or MIM, is defined as
	\begin{equation} \label{eq}
		L(\emph{p},\bar{\omega})=\log  \sum_{i=1}^N p_i\exp(\bar{\omega}(1-p_i)).
	\end{equation}
\end{definition}

Compared with the Shannon entropy, MIM utilizes the exponential
function to replace the logarithmic function to emphasize rare events. It has been proven that the exponential function with different properties from the logarithmic function makes differences in information characterization\cite{r24,r34,r35}. Inspired by this, MIM-based GAN in \cite{r16} introduces the exponential function into the original GAN that contains a logarithmic function in the objective function. Specifically, it is revealed in \cite{r38} that MIM-based GAN has better mode collapse resistance, discriminator stability and higher small probability emphasizing than GAN with the logarithmic funcion. Moreover, from TABLE II in \cite{r38} and TABLE II in \cite{r39}, MIM-based GAN generates more stable small probability events compared with GAN with the logarithmic function. Hence, it is promising to leverage MIM-based GAN to offer a better mutual information estimator.
The objective function of MIM-based GAN is given by
\begin{equation}\label{eq.L_MIM_GD}
		L_{\mathrm{\rm MIM}}(\theta,D) =\mathbb{E}_{\mathbf{x}\sim p}[\exp(1-D(\mathbf{x}))]+\mathbb{E}_{\mathbf{x}\sim p_{\theta}}[\exp(D(\mathbf{x}))],
\end{equation}
in which $p$ and $p_{\theta}$ denote the distributions of real data and generative data, respectively.  Likewise, $D$ is the discriminator and $\theta$ is the parameter group for the neural network of the generator. Its optimal discriminator is given by
\begin{equation}\label{eq.D*_MIM}
	D_{\mathrm{\rm MIM}}^*(\mathbf{x}) = \frac{1}{2}+\frac{1}{2}\ln\frac{p(\mathbf{x})}{p_{\theta}(\mathbf{x})}.
\end{equation}

Similarly to iDIME and dDIME, we design a mutual information estimator originated from MIM-based GAN by replacing $p(\cdot)$ and $p_{\theta}(\cdot)$ with $p_{XY}(\cdot,\cdot)$ and $p_X(\cdot)p_Y(\cdot)$, which is discussed in details as follows.

\subsection{Non-Parametric MIM-Based Mutual Information Estimator}\label{AA}
By resorting to the objective function of MIM-based GAN, the following Definition \ref{def.J_I_MIM} and \ref{def.I_MMIE_D^*} are given to describe the mutual information estimator $I_{\text{\rm MMIE}}(X;Y)$.
\begin{definition}\label{def.J_I_MIM}
Let $X, Y\sim p_{XY}(\mathbf{x},\mathbf{y})$. A value function $J_{\text{\rm MIM}}(D)$ is defined as
	\begin{equation}\label{eq.J_MIM_D}
		\begin{split}
			J_{\text{\rm MIM}}(D)=&\mathbb{E}_{(\mathbf{x},\mathbf{y})\sim p_{XY}(\mathbf{x},\mathbf{y})}[\exp(1-D(\mathbf{x},\mathbf{y}))] \\
			&+\mathbb{E}_{(\mathbf{x},\mathbf{y})\sim p_X(\mathbf{x})p_Y(\mathbf{y})}[\exp(D(\mathbf{x},\mathbf{y}))].
		\end{split}
	\end{equation}
    as well as an estimator function is given by
    \begin{equation}\label{eq.MMIE}
        {\widetilde I}_{\text{\rm MMIE}}(X,Y,D)= 2 \mathbb{E}_{(\mathbf{x},\mathbf{y})\sim p_{XY}(\mathbf{x},\mathbf{y})}[D(\mathbf{x},\mathbf{y})] - 1,
    \end{equation}
    where $D$ is a function $\mathcal{X} \times \mathcal{Y} \to R$ regarded as a discriminator, $p_X(\mathbf{x})=\int_{\mathcal{Y}} p_{XY}(\mathbf{x},\mathbf{y}) {\rm d}\mathbf{y} $ and $p_Y(\mathbf{y})=\int_{\mathcal{X}} p_{XY}(\mathbf{x},\mathbf{y}) {\rm d}\mathbf{x}$, as well as $\mathcal{X}$ and $\mathcal{Y}$ are domains for $X$ and $Y$, respectively.
\end{definition}

\begin{definition}\label{def.I_MMIE_D^*}
    Based on Definition \ref{def.J_I_MIM}, a non-parametric MIM-based Mutual Information Estimator (MMIE) is defined as
    \begin{equation}\label{eq.I_MMIE_D^*}
		\begin{split}
			I_{\text{\rm MMIE}}(X,Y)
            &= {\widetilde I}_{\text{\rm MMIE}}(X,Y,D^*_{\text{\rm MMIE}}) \\
            &=2 \mathbb{E}_{(\mathbf{x},\mathbf{y})\sim p_{XY}(\mathbf{x},\mathbf{y})}
            [D^*_{\text{\rm MMIE}}(\mathbf{x},\mathbf{y})] - 1,
		\end{split}	
	\end{equation}
    where
        \begin{equation}\label{eq.MMIE_D^*}
        \begin{aligned}
    		D^*_{\text{\rm MMIE}}
            & = \mathop{\arg\min}\limits_{D} J_{\text{\rm MIM}}(D).
        \end{aligned}
    \end{equation}
\end{definition}
	The following theorem shows the conditions for the estimator $I_{\text{\rm MMIE}}$ to reach the real value of the mutual information.

\begin{theorem}\label{thm.I_MMIE}
For two random variables $X$ and $Y$ whose probability density functions are Lebesgue integrable, i.e.,
\begin{subequations}
\begin{align}
	& (L)\int_x f(x)dx< \infty, \\
	& (L)\int_y f(y)dy< \infty, \\
	& (L)\int_y\int_x f(x,y)dxdy< \infty,
\end{align}
\end{subequations}
in which $(L)$ indicates the Lebesgue integral,
 we have
    \begin{equation}\label{eq.I_MMIE}
		\begin{split}
			I(X,Y)
            & = I_{\text{\rm MMIE}}(X,Y) \\
            & = {\widetilde I}_{\text{\rm MMIE}}(X,Y,D^*_{\text{\rm MMIE}}),
		\end{split}	
	\end{equation}
where
    \begin{equation}\label{eq.MMIE_D^*}
        \begin{aligned}
    		D^*_{\text{\rm MMIE}}(\mathbf{x},\mathbf{y})
            & = \frac{1}{2}+\frac{1}{2} \ln\frac{p_{XY}(\mathbf{x},\mathbf{y})}{p_X(\mathbf{x})p_Y(\mathbf{y})}.
        \end{aligned}
    \end{equation}
\end{theorem}
\begin{proof}
Please refer to Appendix \ref{app.MMIE}.
\end{proof}
\begin{remark}\label{rem.1}
	In practice, we obtain $D^*_{\text{\rm MMIE}}$ in the family of the functions $D_{\phi}:\mathcal{X} \times \mathcal{Y} \to R$ parametered
	by a deep neural discriminator network with parameters $\phi \in \Phi$, i.e.,
	\begin{equation}
		D^*_{\text{\rm MMIE},\Phi}=\mathop{\arg\min}\limits_{\phi \in \Phi} J_{\text{\rm MIM}}(D_{\phi}).
	\end{equation}
	While, in theory, we usually consider the discriminator network $D_{\phi}$ with enough capacity, i.e., in the non-parametric limit\cite{r23,r9}. That means the parametric network $D_{\phi}$ can converge to the non-parametric function in Eq. (\ref{eq.MMIE_D^*}). Likewise, there exists $\phi^*\in \Phi$, based on which for any $\mathbf{x},\mathbf{y}\sim p_{XY}(\mathbf{x},\mathbf{y})$, we have
	\begin{equation}
		D_{\phi^*}(\mathbf{x},\mathbf{y})=\frac{1}{2}+\frac{1}{2} \ln\frac{p_{XY}(\mathbf{x},\mathbf{y})}{p_X(\mathbf{x})p_Y(\mathbf{y})},
	\end{equation}
	which leads to
	\begin{equation}
		D^*_{\text{\rm MMIE}}=D^*_{\text{\rm MMIE},\Phi}.
	\end{equation}
	Therefore, we omit the network parameter $\phi$ in the discriminator $D_{\phi}$ in the following discussion and assume $D_{\phi}$ can represent $D^*_{\text{\rm MMIE}}$ in Eq.  (\ref{eq.MMIE_D^*}).
\end{remark}

Compared with dDIME, MMIE improves the discriminator output from $\frac{p_{XY}(\mathbf{x},\mathbf{y})}{p_X(\mathbf{x})p_Y(\mathbf{y})}$ to $\ln\frac{p_{XY}(\mathbf{x},\mathbf{y})}{p_X(\mathbf{x})p_Y(\mathbf{y})}$, leading to more accurate estimation when $p_{XY}(\mathbf{x},\mathbf{y}) \ll p_X(\mathbf{x})p_Y(\mathbf{y})$. Specifically, it is revealed in \cite{r9} that when the pointwise mutual information is low, i.e, $p_{XY}(\mathbf{x},\mathbf{y}) \ll p_X(\mathbf{x})p_Y(\mathbf{y})$, the output of the discriminator $D_{\text{\rm dDIME}}(\mathbf{x},\mathbf{y})$ may saturate to 0, causing estimation error or even numerical failure (i.e., the logarithm of zero tends to infinity). Such a phenomenon is also observed in our experiments in Section VA. By contrast, with the same settings, the output of the discriminator $D_{\text{\rm MMIE}}(\mathbf{x},\mathbf{y})$ is always bounded in our experiments.
Hence, MMIE is more suitable for mutual information estimation. 
	
\begin{remark}
	The difference between d-DIME and MMIE lies in not only the mathematical expressions for the discriminator output, but also the network construction. For example, since the positivity of the output, dDIME adopts softplus as the activation function in the last layer. However, the softplus is not scale-independent, which partly leads to the usage of the parameter $\alpha$. While, MMIE adopts a linear function for the last layer, which is scale independent.
\end{remark}

\subsubsection{Convergence}
	The following theorem guarantees that the convergence of MMIE is controlled by the convergence of the discriminator $D$, which tends to reach the optimal solution when minimizing the value function $J_{\text{MIM}}(D)$.
\begin{prop}\label{prop.convergence_JMIM}
	Let the discriminator $D$ be with enough capacity, i.e., in the non-parametric limit, as mentioned in Remark \ref{rem.1}. Consider  $D^*_{\text{\rm MMIE}}=\mathop{\arg\min}_{D} J_{\text{\rm MIM}}(D)$ where $J_{\text{\rm MIM}}(\cdot)$ corresponds to Eq. (\ref{eq.J_MIM_D}), and utilize the gradient descent method as the update rule for $D$ to achieve $D^*_{\text{\rm MMIE}}$.
	Then, the estimator function ${\widetilde I}_{\text{\rm MMIE}}(X,Y,D)$ defined in Eq. (\ref{eq.MMIE})
    converges to the mutual information $I(X,Y)$, as well as the convergence rate of ${\widetilde I}_{\text{\rm MMIE}}$ is in proportion to the convergence of the discriminator $D$.
\end{prop}
\begin{proof}
    From the gradient descent method\cite{r36}, i.e. $D^{(n+1)} = D^{(n)} -\mu\nabla J_{\text{MIM}}(D^{(n)})$, we have
	\begin{equation}
		\begin{split}
			& {\widetilde I}^{(n+1)}_{\text{MMIE}}- {\widetilde I}^{(n)}_{\text{MMIE}} \\
            &= 2\mathbb{E}_{(\mathbf{x},\mathbf{y})\sim p_{XY}(\mathbf{x},\mathbf{y})}[D^{(n+1)}-D^{(n)}]  \\
			&=-2\mu\mathbb{E}_{(\mathbf{x},\mathbf{y})\sim p_{XY}(\mathbf{x},\mathbf{y})}[\nabla J_{\text{MIM}}(D^{(n)})].
		\end{split}		
	\end{equation}
Also, we have 
	\begin{equation}
	\begin{split}
		& {\widetilde I}^{(n)}_{\text{MMIE}}- I_{\text{MMIE}} \\
		&= 2\mathbb{E}_{(\mathbf{x},\mathbf{y})\sim p_{XY}(\mathbf{x},\mathbf{y})}[D^{(n)}-D^*_{\text{MMIE}}]  \\
	\end{split}		
\end{equation}

    If $D$ converges to the global optimal solution under the gradient descent method, then the mutual information estimator ${\widetilde I}_{\text{MMIE}}$ converges to $I_{\text{MMIE}}$ which is equal to the real value of the mutual information $I(X,Y)$.
	Moreover, since the expectation is a linear operator and does not affect convergence, there exists a proportional relationship between the convergence rate for $I_{\text{MMIE}}$ and that for the discriminator $D$.
\end{proof}

\subsubsection{Relationship with R\'{e}nyi Divergence}
Though the cost function $J_{\text{MIM}}(D)$ is not directly related to the mutual information or the K-L divergence like d-DIME, it has the relationship with R\'{e}nyi divergence which is defined as
	\begin{equation}
		\mathcal{R}_{\gamma}(P\|Q)=\frac{1}{\gamma-1}\ln\mathbb{E}_{\mathbf{x}\sim P}\left[\left(\frac{P(\mathbf{x})}{Q(\mathbf{x})}\right)^{\gamma-1}\right], \quad(\gamma>0)
	\end{equation}
where $P$ and $Q$ are two distributions.

\begin{prop}\label{prop.MMIE_Renyi}
Let $J_{\text{\rm MIM}}(D^*_{\text{\rm MMIE}})=\mathop{\min}_{D} J_{\text{\rm MIM}}(D)$ denote the minimum of the value function $J_{\text{\rm MIM}}$ with respect to the discriminator $D$.
We have
	\begin{equation}
		J_{\text{\rm MMIE}}(D^*_{\text{\rm MMIE}}) = 2\exp\left(\frac{1}{2}(1-\mathcal{R}_{\frac{1}{2}}(p_{XY}\|p_Xp_Y))\right),
	\end{equation}
	where $\mathcal{R}_{\frac{1}{2}}(\cdot)$ is the R\'{e}nyi divergence with $\gamma=\frac{1}{2}$, $p_{XY}(\mathbf{x},\mathbf{y})$ is the joint distribution of $X$ and $Y$, and $p_X(\mathbf{x})p_Y(\mathbf{y})$ is the product of the marginal distributions. That is,
	\begin{equation}\label{eq.R_JMIM}
		\begin{split}
			\mathcal{R}_{\frac{1}{2}}(p_{XY}\|p_Xp_Y)&=-2\ln(J_{\text{\rm MIM}}(D^*_{\text{\rm MMIE}}))+1+2\ln2 \\
										  &\ge-2\ln(J_{\text{\rm MIM}}(D))+1+2\ln2.
		\end{split}
	\end{equation}
\end{prop}
\begin{proof}
please refer to Appendix \ref{app.MMIE_Renyi}.
\end{proof}

The inequality in (\ref{eq.R_JMIM}) suggests that for an arbitrary discriminator function $D(\mathbf{x},\mathbf{y})$, $-2\ln(J_{\text{\rm MIM}}(D))+1+2\ln2$ is a lower bound for the R\'{e}nyi divergence $\mathcal{R}_{\frac{1}{2}}(p_{XY}\|p_Xp_Y)$.

\begin{remark}
Mutual information can be expressed as a form of the K-L divergence, which is a special case of the R\'{e}nyi divergence when the parameter is set as $\gamma=1$. Hence, we can consider the R\'{e}nyi divergence as an information metric to measure the dependence between random variables, and maximize it by minimizing $J_{\text{\rm MIM}}(D^*)$.
If we define the channel capacity as the maximum R\'{e}nyi divergence between the channel input and output, we can derive a cooperative framework to estimate such R\'{e}nyi capacity which is given by
\begin{equation}\label{eq.Renyi_capacity}
	\begin{split}
		C_R&=\mathop{\max}\limits_{p_X(\mathbf{x})} \mathcal{R}_{\frac{1}{2}}(p_{XY}\|p_Xp_Y) \\
		&=\mathop{\max}\limits_{p_X(\mathbf{x})}\left[-2\ln(J_{\text{\rm MIM}}(D^*_{\text{\rm MMIE}}))\right]+1+2\ln2 \\
		&=\mathop{\max}\limits_{p_X(\mathbf{x}),D}\left[-2\ln(J_{\text{\rm MIM}}(D))\right],
	\end{split}
\end{equation}
from which the optimal input distribution is also obtained.
\end{remark}

\subsection{{$\alpha$}-Parameterized MIM-based Mutual Information Estimator}

The estimators based on objective functions of GAN suffer from performance deterioration in the case of high SNR, especially for dDIME. The reason is that these kinds of GAN are based on the assumption that the output of the discriminator converges to a known constant, which is difficult to achieve under a high SNR.

Let us take the original GAN as an example. The objective function and optimal output for the  discriminator are given by
\begin{equation}\label{eq.L_GAN_GD}
	L_{\text{GAN}}(G,D)=\mathbb{E}_{\mathbf{x}\sim p}[\ln(1-D(\mathbf{x}))]+\mathbb{E}_{\mathbf{x}\sim p_{\theta}}[\ln(D(\mathbf{x}))],
\end{equation}
and
\begin{equation}
	D^*_{\text{GAN}}(\mathbf{x})=\frac{1}{1+\frac{p_{\theta}(\mathbf{x})}{p(\mathbf{x})}},
\end{equation}
where $p(\cdot)$ and $p_{\theta}(\cdot)$ are the real and generative distributions, $D$ is the discriminator and $G$ is the generator with the parameter group $\theta$.
During the adversarial training process, the whole network makes $p_{\theta}(\cdot)$ approach to $p(\cdot)$ so that the optimal discriminator output $D^*(\mathbf{x})$ converges to $\frac{1}{2}$. At the equilibrium point, it is easy to see that the two terms in Eq. (\ref{eq.L_GAN_GD}) give equal contributions to the total objective function. This phenomenon also happens to MIM-based GAN, which is seen from Eq. (\ref{eq.L_MIM_GD}) and (\ref{eq.D*_MIM}).

However, when these objective functions are utilized for mutual information or capacity estimation, the output of the optimal discriminator may be influenced by the input $\mathbf{x}$, or even the distribution of $D^*(\mathbf{x})$ may be with a large deviation from the prior settings.
With respect to MMIE under a high SNR, the expectation of $D^*_{\text{\rm MIM}}(\mathbf{x})$ is given by
\begin{equation}
	\mathbb{E}_{\mathbf{x},\mathbf{y}\sim p_{XY}(\mathbf{x},\mathbf{y})}[D^*_{\text{MIM}}]=\frac{1}{2}I(X,Y)+\frac{1}{2}\gg \frac{1}{2},
\end{equation}
which indicates that the second term in Eq. (\ref{eq.L_MIM_GD}) has much more contributions than the first term.
This implies the results are unstable during the training process.
Hence, a parametric form of MMIE, named $\alpha- \text{MMIE}$, is designed to mitigate the deviation issue.

\begin{definition}\label{def.J_alpha_MIM}
Consider two random variables $X$ and $Y$ with the joint distribution $p_{XY}(\mathbf{x},\mathbf{y})$, then there is a value function $J_{\alpha,\text{\rm MIM}}(D)$ defined as
	\begin{equation}\label{eq.J_alpha_MIM}
		\begin{split}
			J_{\alpha,\text{\rm MIM}}(D)=&\mathbb{E}_{(\mathbf{x},\mathbf{y})\sim p_{XY}(\mathbf{x},\mathbf{y})}[\exp(\alpha-D(\mathbf{x},\mathbf{y}))] \\
			&+\mathbb{E}_{(\mathbf{x},\mathbf{y})\sim p_X(\mathbf{x})p_Y(\mathbf{y})}[\exp(D(\mathbf{x},\mathbf{y})-\alpha)],
		\end{split}
	\end{equation}
as well as the corresponding estimator function is given by
    \begin{equation}\label{alpha_MMIE}
    {\widetilde I}_{\alpha \text{\rm MMIE}}(X,Y,D)
    = 2 \mathbb{E}_{(\mathbf{x},\mathbf{y})\sim p_{XY}(\mathbf{x},\mathbf{y})}[D(\mathbf{x},\mathbf{y})] -2\alpha.
    \end{equation}
Here $\alpha$ is an adjustable parameter, $D$ is a function $\mathcal{X} \times \mathcal{Y} \to R$ regarded as a discriminator, $p_X(\mathbf{x})=\int_{\mathcal{Y}} p_{XY}(\mathbf{x},\mathbf{y}) {\rm d}\mathbf{y} $ and $p_Y(\mathbf{y})=\int_{\mathcal{X}} p_{XY}(\mathbf{x},\mathbf{y}) {\rm d}\mathbf{x}$. Likewise, $\mathcal{X}$ and $\mathcal{Y}$ are the domains for $X$ and $Y$, respectively.
\end{definition}

\begin{definition}\label{def.alpha_I_MMIE_D^*}
    Based on Definition \ref{def.J_alpha_MIM}, the $\alpha$-parameterized MIM-based Mutual Information Estimator ($\alpha-\text{MMIE}$) is defined as
    \begin{equation}\label{eq.alpha_I_MMIE_D^*}
		\begin{split}
			{I}_{\alpha \text{\rm MMIE}}(X,Y)
            &= {\widetilde I}_{\alpha \text{\rm MMIE}}(X,Y,D^*_{\alpha \text{\rm MMIE}}) \\
            &=2 \mathbb{E}_{(\mathbf{x},\mathbf{y})\sim p_{XY}(\mathbf{x},\mathbf{y})}
            [D^*_{\alpha \text{\rm MMIE}}(\mathbf{x},\mathbf{y})] - 2\alpha,
		\end{split}	
	\end{equation}
    where
        \begin{equation}\label{eq.alpha_MMIE_D^*}
        \begin{aligned}
    		D^*_{\alpha \text{\rm MMIE}}
            & = \mathop{\arg\min}\limits_{D} J_{\alpha,\text{\rm MIM}}(D).
        \end{aligned}
    \end{equation}
\end{definition}

\begin{theorem}\label{thm.I_alpha_MMIE}
From Definition \ref{def.alpha_I_MMIE_D^*}, for two random variables $X$ and $Y$ whose probability density functions are Lebesgue integrable, i.e.,
	\begin{subequations}
		\begin{align}
			& (L)\int_x f(x)dx< \infty, \\
			& (L)\int_y f(y)dy< \infty, \\
			& (L)\int_y\int_x f(x,y)dxdy< \infty,
		\end{align}
	\end{subequations}
	in which $(L)$ indicates the Lebesgue integral,
the mutual information between $X$ and $Y$ is given by
    \begin{equation}\label{eq.alpha_I_MMIE}
		\begin{split}
			I(X,Y)
            &= I_{\alpha \text{\rm MMIE}}(X,Y) \\
			&=2 \mathbb{E}_{(\mathbf{x},\mathbf{y})\sim p_{XY}(\mathbf{x},\mathbf{y})}[D^*_{\alpha \text{\rm MMIE}}(\mathbf{x},\mathbf{y})] -2\alpha,
		\end{split}	
	\end{equation}
where
    \begin{equation}\label{eq.D_*_alpha_MMIE}
        \begin{aligned}
		D^*_{\alpha \text{\rm MMIE}}(\mathbf{x},\mathbf{y})
        & =\frac{1}{2} \ln\frac{p_{XY}(\mathbf{x},\mathbf{y})}{p_X(\mathbf{x})p_Y(\mathbf{y})} + \alpha.
        \end{aligned}
    \end{equation}
\end{theorem}
\begin{proof}
Please refer to Appendix \ref{app.MMIE}.
\end{proof}

The parameter $\alpha$ here is similar to the parameter $\alpha$ in dDIME, which adjusts the distribution of the optimal discriminator. When $\alpha=\frac{1}{2}$, the $\alpha-\text{MMIE}$ is equivalent to MMIE.

\subsubsection{Convergence}
The following theorem guarantees that the convergence of $\alpha-$MMIE is controlled by the convergence of the discriminator $D$, where the optimal $D^*$ will be achieved by minimizing the value function $J_{\alpha,\text{MIM}}(D)$.

\begin{prop}\label{prop.convergence_J_alpha_MIM}
	Consider the discriminator $D$ without the parametric constraint. Let  $D^*_{\alpha \text{\rm MMIE}}=\mathop{\arg\min}\limits_{D} J_{\alpha,\text{\rm MIM}}(D)$ where $J_{\alpha,\text{\rm MIM}}(\cdot)$ is given by Eq. (\ref{eq.J_alpha_MIM}), and use the gradient descent method as the update rule for $D$.
	Then, the mutual information estimator ${\widetilde I}_{\alpha \text{\rm MMIE}}(X,Y,D)$ as given by Eq. (\ref{alpha_MMIE})
    converges to the mutual information $I(X,Y)$.
    In addition, there exists a proportional relationship between the convergence rate of ${\widetilde I}_{\alpha \text{\rm MMIE}} $ and that for the discriminator $D$.
\end{prop}

Similar to Proposition \ref{prop.convergence_JMIM}, the proof of the above proposition can be easily obtained.

\subsubsection{Selection of the Parameter $\alpha$}

We provide an empirical guideline for selecting $\alpha$ by setting the mean of the discriminator output as zero, i.e.,
\begin{equation}
	\mathbb{E}_{\mathbf{x},\mathbf{y}\sim p_{XY}(\mathbf{x},\mathbf{y})}[D^*_{\alpha \text{\rm MMIE}}(\mathbf{x},\mathbf{y})]=0,
\end{equation}
from which, we have
\begin{equation}
	\alpha = -\frac{1}{2}I(X,Y).
\end{equation}
Such a setting tends to ensure the zero-centered output of the discriminator, which is good for drift or error neutralization in practical applications.

Although the exact value of $I(X,Y)$ is unknown in advance, a rough approximation $\hat{I}(X,Y)$ is easily pre-estimated. Specifically, we assume the channel is AWGN, and the power of noise can be obtained through Monte Carlo methods. Then, it follows that
\begin{equation}
	I(X,Y)\le R_1=\frac{d}{2}\ln(1+\sigma^{-2}),
\end{equation}
On the other hand, if the channel input is discrete, for example, with $M$ possible messages, we have
\begin{equation}
	I(X,Y)\le R_2 =\ln M.
\end{equation}

In this regard, a rough estimation can be given as
\begin{equation}
	\hat{I}(X,Y)=\beta \min (R_1,R_2),
\end{equation}
where $\beta$ is a positive parameter less than 1 (we take $\beta=0.7$ in the experiments).

\section{Cooperative Network for Capacity Learning}
In this section, we propose a cooperative framework based on MMIE to learn the channel capacity and design the codebook simultaneously.

First of all, we generalize the framework named 
Cooperative  Networks for Capacity Learning (CORTICAL) in \cite{r9} for channel capacity learning. Similar to the original GAN, the framework in Fig. \ref{fig_1} consists of a pair of networks, namely, the discriminator and the generator. The generator is an encoder that produces the channel input and is also regarded as the source codebook output, while the discriminator is a mutual information estimator for channel capacity learning.
Specifically, through adversarial training, the discriminator learns the value of the channel capacity, while the generator achieves the optimal input distribution model that reaches the channel capacity.
In addition, the data augmentation module is regarded as a bridge between the two networks.

Different from CORTICAL, in this generalized framework, the discriminator is obtained from any differential sample-based estimators rather than a particular discriminative estimator.
For example, the discriminator can be a generative estimator like BA estimator, then the data augmentation module produces samples from the distributions $p_X(\cdot)$, $p_Y(\cdot)$ and $p_{XY}(\cdot,\cdot)$ based on input $\mathbf{x}$. While, if the discriminator is a discriminative estimator, the data augmentation module shuffles the channel output $H(\mathbf{x})$ to get paired and unpaired samples from $p_{XY}(\cdot,\cdot)$ and $p_X(\cdot)p_Y(\cdot)$, respectively.
In addition, the generalized framework allows the generator and discriminator to adopt different objective functions. In this regard, the framework incorporates other estimators such as iDIME and MMIE, whose objective functions are not directly related to the mutual information.

\begin{figure}[htbp]
	\centering
    \includegraphics[scale=0.43]{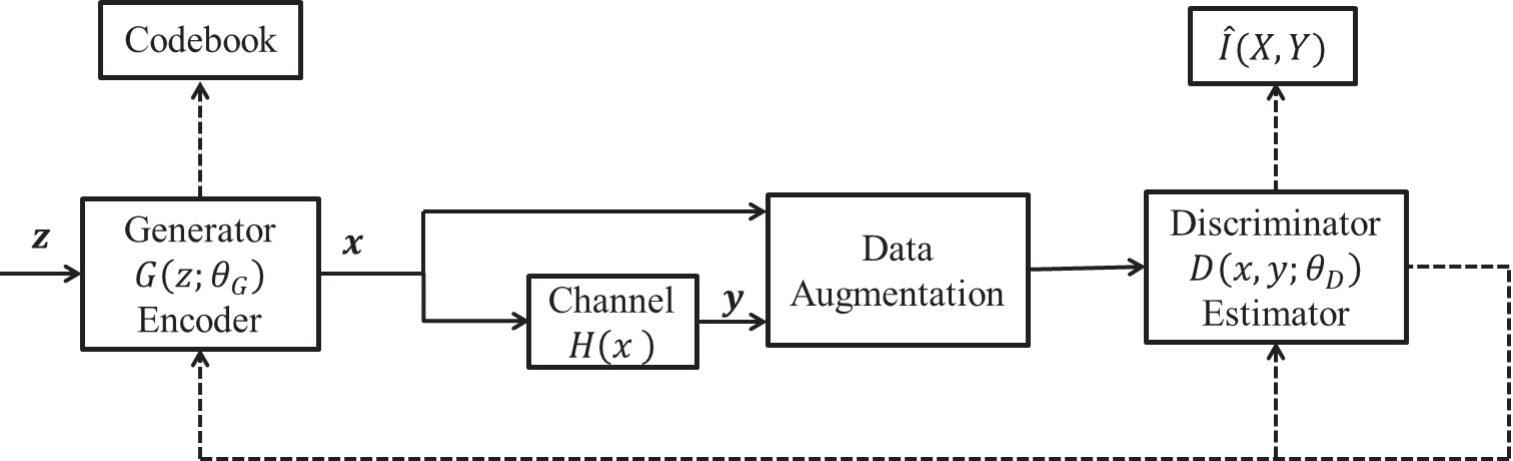}\\
	\caption{Generalized cooperative framework for capacity learning. The generator produces samples from the channel input distribution, as well as the discriminator estimates the mutual information based on the input.}
	\label{fig_1}
\end{figure}


\begin{assum}\label{ass.alpha_channel}
Consider a stochastic channel denoted by $H(\cdot)$, and let $X\sim p_X(\mathbf{x})$ and $ Y\sim p_{Y}(\mathbf{y})$ be its input and output. Then, $Y=H(X)$. Let $\pi(\cdot)$ denote the permutation function, such that
$p_Y(\mathbf{y}) = p_{\pi(Y)|X}(\pi(\mathbf{y})|\mathbf{x})$.
\end{assum}

In practice, the exact form of permutation function $\pi(\cdot)$ is usually hard to obtain. Hence, it is replaced by a shuffle operation. Specifically, for a set of input-output sample pairs $(\mathbf{x}_i,\mathbf{y}_i),i=1,\cdots,n$, we randomly permute $\mathbf{y}_i$ to obtain $\mathbf{y}_{i'}$, so that these sample pairs become $(\mathbf{x}_i,\mathbf{y}_{i'})$.
Here, $(\mathbf{x}_i,\mathbf{y}_{i'})$ is viewed as $(\mathbf{x},\pi(\mathbf{y}))$.

Suppose Assumption \ref{ass.alpha_channel} holds,
as well as let $\mathbf{x}= G(\mathbf{z};\theta_G)$ and $\mathbf{y}=H(G(\mathbf{z};\theta_G))$, where $G(\cdot;\theta_G): \mathcal{Z}\to\mathcal{X}$ is a family of functions parametrized
by a deep neural generator network with parameters $\theta_G \in \Theta_G$, and $\mathbf{z}$ is a random variable following a given distribution $p_Z(\mathbf{z})$. Besides, $D(\cdot,\cdot;\theta_D)$ in Fig \ref{fig_1} is a family of functions, i.e., $D(\cdot,\cdot;\theta_D): \mathcal{X}\times\mathcal{Y}\to R$, parametrized by a deep neural discriminator network with parameters $\theta_D \in \Theta_D$. For simplicity, we omit $\theta_D,\theta_G$ in $D$ and $G$ similar to those in \cite{r23,r9,r16}.
Then, with respect to MMIE, we have the value function as follows:
	\begin{equation}\label{eq.J_MIM_G_D}
		\begin{split}
			{\hat J}_{\text{MIM}}(G,D)=&\mathbb{E}_{\mathbf{z}\sim p_{Z}(\mathbf{z})}[\exp(1-D(G(\mathbf{z}),H(G(\mathbf{z}))))] \\
			&+\mathbb{E}_{\mathbf{z}\sim p_Z(\mathbf{z})}[\exp(D(G(\mathbf{z}),\pi(H(G(\mathbf{z})))))].
		\end{split}
	\end{equation}
    Likewise, the mutual information estimator is given by
	\begin{equation}\label{eq.widetilde_I_MMIE_G_D}
		\hat {I}_{\text{MMIE}}(G,D)=2\mathbb{E}_{\mathbf{z}\sim p_Z(\mathbf{z})}[D(G(\mathbf{z}),H(G(\mathbf{z})))]-1.
	\end{equation}

\begin{theorem}
    Suppose that Assumption \ref{ass.alpha_channel} holds, then the channel capacity $C$ is given by
	\begin{equation}
		C=\mathop{\max}\limits_{G} {\hat I}_{\text{\rm MMIE}}(G,D^*_G),
	\end{equation}
	in which
	\begin{equation}
		D^*_G =\mathop{\arg\min}\limits_{D}{\hat J}_{\text{\rm MIM}}(G,D),
	\end{equation}
	with the same generator $G$ in ${\hat I}_{\text{\rm MMIE}}(G,D^*_G)$.
    Likewise, ${\hat J}_{\text{\rm MIM}}(G,D)$ and ${\hat I}_{\text{\rm MMIE}}(G,D^*_G)$ are obtained from Eq. (\ref{eq.J_MIM_G_D}) and Eq. (\ref{eq.widetilde_I_MMIE_G_D}), respectively.
\end{theorem}
This theorem is easily proved by resorting to the definition of the channel capacity \cite{cover1999elements}.
Similarly, it is also available for $\alpha-\text{MMIE}$ to achieve the channel learning, whose procedure is the same as that mentioned above.

In practice, the random variable $\mathbf{z}$ can be considered as signals from a certain source, with the distribution $p_Z(\mathbf{z})$. Then, the generator $G$ maps $\mathbf{z}$ into code $\mathbf{x}$ with the distribution $p_X(\mathbf{x})$. Through the channel model $H(\cdot)$, we get the channel output $H(\mathbf{x})$, and by shuffling the output, samples from marginal distribution are also derived. The discriminator and generator adjust their parameters cooperatively to reach the optimal point.

If the source $\mathbf{z}$ is discrete with finite input alphabet of dimension $M$, the optimal input distribution becomes codebook and is easy to derive by sending all the $M$-dimension $\mathbf{z}$ into the generator already trained, then the $M-$dimension output of the generator is the codebook approaching the channel capacity.

\section{Simulations}
In this section, we provide some simulations to compare the performance of MMIE and other estimators for the mutual information.
The AWGN channel is recommended to be a valuable test channel\cite{r9}, since it has closed-form mutual information when the channel input is Gaussian, and this closed-form channel capacity depends on the SNR.

\subsection{Accuracy of the Mutual Information Estimation}\label{section.accuracy_MI}

First of all, we compare the performance of MMIE with other mutual information estimators. For traditional methods, we choose the Kraskov-St{\"o}gbauer-Grassberger (KSG) estimator based on the k-nearest neighbor density estimation\cite{r37}. Like other traditional methods, KSG is not adaptive to large sample sizes and high dimensions due to its high computational complexity. Hence, we can barely compare KSG with neural estimators with the same sample size, either too small for neural estimators to be trained or too large for KSG to compute. Therefore, we choose $N=10000$ as suggested in \cite{r37} that the KSG is suitable for $N$ up to a few thousand.

For neural estimators, we choose MINE, iDIME and dDIME. The parameter of dDIME is only set as $\alpha=0.1$, since high values of $\alpha$ will cause a numerical failure for a high SNR. For simplicity, a two-hidden-layers multilayer perceptron (MLP) neural network is adopted as the backbone for all the estimators, with different activation functions for the additional third hidden layer and the output layer. All neural estimators are trained for $10000$ 
iterations. Further details of the network architecture and parameter settings are given in Appendix B.
	
Consider a d-dimensional AWGN channel as follows
\begin{equation}
	Y=X+N,
\end{equation}
or
\begin{equation}
	Y_i = X_i+N_i, i =1,2,\cdots,d.
\end{equation}
Let $X\sim \mathcal{N}(0,\mathbf{I})$ and $N\sim \mathcal{N}(0,\sigma^2 \mathbf{I})$ where $\mathcal{N}$ denotes the Gaussian distribution, $\mathbf{I}$ is the identity matrix, $\sigma$ is the parameter for the covariance. Then, the mutual information of $X$ and $Y$ is derived as
\begin{equation}\label{eq.I_estimation}
	I=\frac{d}{2}\ln(1+\sigma^{-2}).
\end{equation}
In particular, we choose $d=2$ and $d=10$ as the low and high dimensional cases for our simulations, respectively.

We train $100$ estimators for all the neural approaches and each estimator receives $1000$ testing sets with a batch of $512$ samples. Hence, for each method, $100000$ estimated values of the mutual information are obtained. For the KSG method, we run it $10000$ times with sample size $N=10000$ to obtain the estimated values. Besides, the test sample size for neural estimators is much smaller than KSG. In this case, the variance of neural estimators is large enough to clearly show the comparison results. Moreover, this also leads to a variance decline for KSG. Nevertheless, KSG is less effective than neural estimators even with less variance. The simulation results are illustrated in Figs. \ref{fig_2} and \ref{fig_3}.

\begin{figure}[!htb]
	\centering
    \includegraphics[scale=0.48]{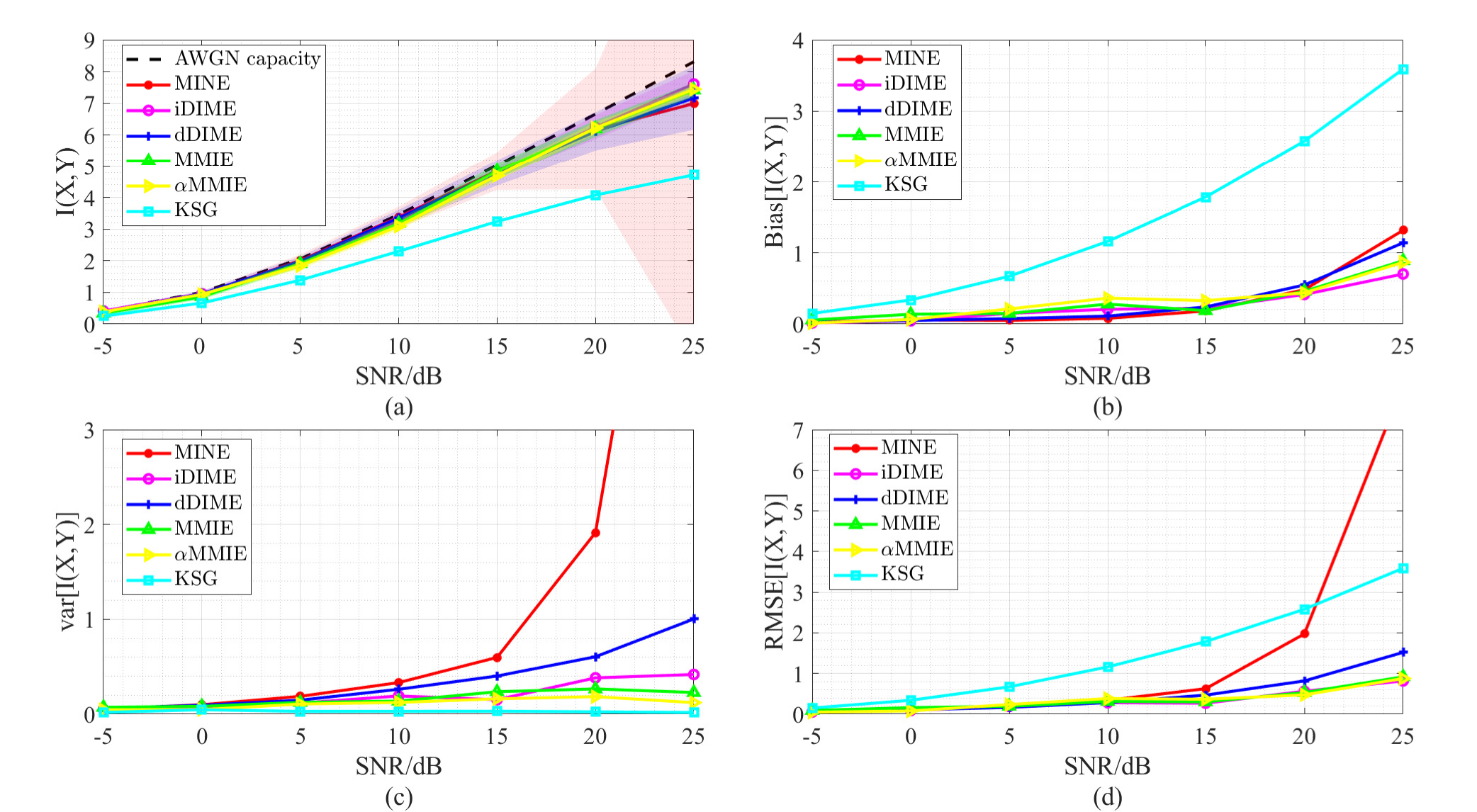}\\
	\caption{ Performance of the mutual information estimators in the case of $2$-dimensional Gaussian variables. (a) Estimated Mutual information $I(X,Y)$ versus SNR; (b) Bias of the estimated mutual information $\text{\rm Bias}[I(X,Y)]$ versus SNR; (c) Variance of the estimated mutual information $\text{\rm var}[I(X,Y)]$ versus SNR; (d) RMSE of the estimated mutual information $\text{\rm RMSE}[I(X,Y)]$ versus SNR. 
}
	\label{fig_2}
\end{figure}

\begin{figure}[!htb]
	\centering
    \includegraphics[scale=0.42]{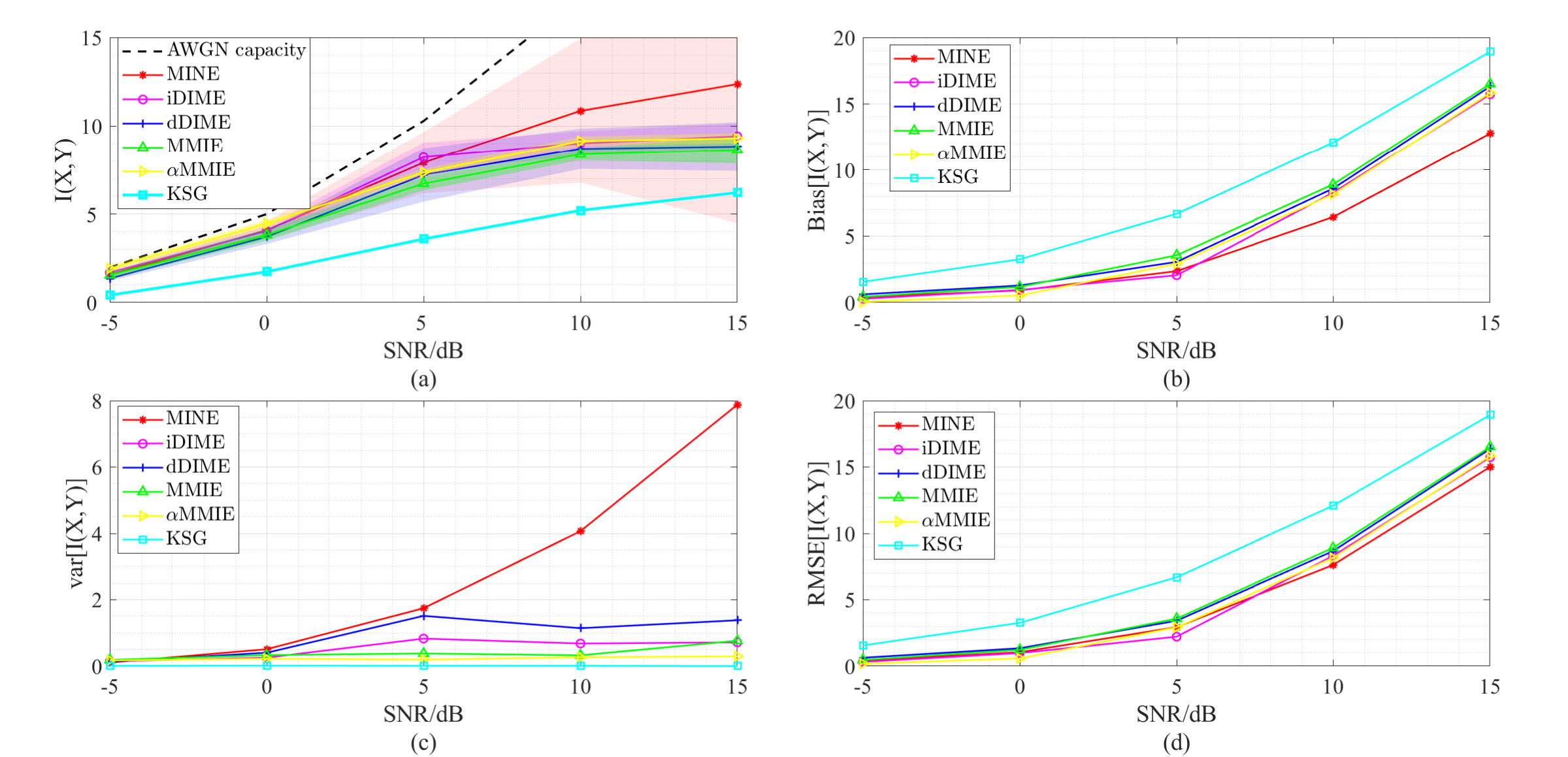}\\
	\caption{ Performance of the mutual information estimators in the case of $10$-dimensional Gaussian variables. Estimated Mutual information $I(X,Y)$ versus SNR; (b) Bias of the estimated mutual information $\text{\rm Bias}[I(X,Y)]$ versus SNR; (c) Variance of the estimated mutual information $\text{\rm var}[I(X,Y)]$ versus SNR; (d) RMSE of the estimated mutual information $\text{\rm RMSE}[I(X,Y)]$ versus SNR. 
}
	\label{fig_3}
\end{figure}
Fig. \ref{fig_2} shows the performance for different estimators with respect to $2$-dimensional Gaussian variables, i.e. $d=2$.  We observe that the bias of KSG is much larger than neural estimators, especially when the SNR is high. In this regard, neural methods outperform traditional methods with large sample sizes. Despite its low variance, the RMSE of KSG is also the largest except for the MINE with the high SNR. As for comparison among neural estimators, though MINE has the least bias for low SNR, it suffers from a heavy variance, which increases at an approximately exponential rate with SNR.
Compared with other estimators, $\alpha - \text{MMIE}$ has the minimum variance and root mean squared error (RMSE). MMIE has a lower bias than $\alpha - \text{MMIE}$ in the case that $\text{SNR}>5\text{dB}$, which is also with a little more variance.

Fig. \ref{fig_3} illustrates the performance for different estimators with respect to $10$-dimensional Gaussian variables, i.e. $d=10$. Similar to the case with $d=2$, KSG is with the largest bias and RMSE, and MINE suffers from an exponentially increasing variance. However, the difference is that MINE has much less bias in the case of high SNR. Theoretical analysis has shown that MINE is upper-biased for mini-batch samples, which explains its optimal accuracy in some degree. As for the rest of the estimators, $\alpha-\text{MMIE}$ has the best performance both in bias and variance. Besides, MMIE has the second least variance among these estimators.

In the training and testing processes, iDIME suffers from numerical issues. Specifically, when SNR is high enough, the outputs of several iDIME estimators tend to the infinity for certain test batches, which implies that the output of the discriminator degenerates to zero or one. When such failure occurs in the training process, this kind of estimator becomes invalid. For example, when $\text{SNR}=25\text{dB}$, there are only a few test batches with failure for iDIME in the case of $d=2$, while nearly half of the test batches fail in the case of $d=10$. In the experiments, we delete these invalid estimators. To make the testing data size consistent, we repeat the training procedure until $10$ valid estimators are obtained. dDMIE with a large $\alpha$, such as $\alpha=1$, suffers a similar failure in the case of high dimensional data and high SNR. In contrast, MMIE and $\alpha-\text{MMIE}$ are without testing failure.

\subsection{Stability of the Mutual Information Estimation}\label{section.stability_MI}
In this subsection, we provide a more specific evaluation of the stability of these neural estimators through the training failure rate, denoted as $r_e$, and the test failure rate, denoted as $r_s$. Specifically, training failure can also be called estimator failure, which occurs when the estimator encounters a numerical failure, such as outputting infinity or NaN (not a number), for all inputs. The training failure rate represents the ratio of failed estimators to all trained estimators.
On the other hand, test failure can also be referred to as sample failure. In this case, the trained estimator only encounters numerical failure for certain testing batches. Similarly, the test failure rate refers to the ratio of all failed input batches to all input batches. Hence, we have
\begin{subequations}
	\begin{align}
		r_e =& \frac{N_{\text{inv}}}{N_{\text{e}}}, \\
		r_s =&  \frac{B_{\text{fb}}}{B_{\text{b}}}
	\end{align}
\end{subequations}
where $N_{\text{inv}}$ is the number of invalid estimators for all input, $N_{\text{e}}$ is the number of all trained estimators,
$B_{\text{fb}}$ is the number of failed testing batches and $B_{\text{b}}$ is the number of all testing batches.

We train $100$ estimators for all the neural approaches and each estimator receives $1000$ testing sets with a batch of $512$ samples. Other training details are the same as the former subsection. Table \ref{table.5} illustrates the performance of these estimators with respect to $10$-dimensional Gaussian variables, i.e. $d=10$. We observe that when $\text{SNR}=10\text{dB}$, iDIME only encounters a few sample failures, while when $\text{SNR}=15\text{dB}$, nearly one-third of iDIME estimators are invalid. For dDIME with $\alpha=1$, the estimators all encounter numerical issues when $\text{SNR}\ge10\text{dB}$, while even for dDIME with $\alpha=0.1$, nearly half of the estimators meet training failure when $\text{SNR}=10\text{dB}$. By contrast, MINE, MMIE, and $\alpha$-MMIE do not encounter numerical issues in the experiments, suggesting their stability. The results can also show that, unlike iDIME and dDIME whose output degenerates to zero or infinity, the output of the discriminator MMIE is always bounded.
 \begin{table}[htb]
	\centering
\caption{\bf Comparision of the Stability of Different Neural Mutual Information Estimators}\label{table.5}
\begin{tabular}{|p{2.5cm}|p{1cm}<{\centering}|p{1cm}<{\centering}|p{1cm}<{\centering}|p{1cm}<{\centering}|p{1cm}<{\centering}|p{1cm}<{\centering}|p{1cm}<{\centering}|p{1cm}<{\centering}|}
	\hline
	
	\multirow{2}{2.7cm}{Approaches}
	&\multicolumn{2}{|c|}{SNR=0dB}  & \multicolumn{2}{|c|}{SNR=5dB}&\multicolumn{2}{|c|}{SNR=10dB} &\multicolumn{2}{|c|}{SNR=15dB} \\
	\cline{2-9} & $r_e$ & $r_s$& $r_e$ & $r_s$&$r_e$ & $r_s$&$r_e$ & $r_s$\\
	\hline
	MINE & 0 &0 &0 &0&0&0&0&0 \\

	iDIME &0&0& 0& 0 & 0 & 0.078 & 0.32 & 0.32 \\
	dDIME ($\alpha=1$) &0&0& 0.35& 0.35 &1 &1 &1 &1 \\
	dDIME ($\alpha=0.1$) &0&0&0 &0 &0.47&0.47&0.3 &0.3\\
	\hline
	MMIE (ours) & 0 &0&0&0 &0 &0&0&0 \\
	$\alpha$-MMIE (ours)  &0&0& 0 &0 &0 &0&0&0 \\
	\hline
	
\end{tabular}
\end{table}

\subsection{Cooperative Networks Performance}

As discussed in \cite{r9}, we consider a discrete channel input to evaluate the performance of the codebooks from the estimators.
Specifically, the generator encodes $M=8$ possible messages into bidimensional constellation points.

We train the generator and discriminator iteratively. That is, we perform a generator training iteration after $25$ discriminator training iterations. The number of total discriminator training iterations is set as $10000$.

Different from \cite{r9}, it is not sufficient to only use the estimated channel capacity produced by the discriminator to evaluate the performance of the codebook. The inaccuracy lies in two issues: (a) the discriminator itself produces an estimation error as shown in Section \ref{section.accuracy_MI}; and (b) the constraint of the power of the channel input $\mathbf{x}$ is realized through the batch normalization layer in the generator. However, in the case of the mini batch, the power of the samples generated from $G$ may fluctuate around $1$, which makes the real SNR deviates from the preset value.

We shall use a supportive mutual information estimator to accurately evaluate the performance of the discrete codebook. As for the AGWN channel, the estimation of the mutual information is simplified as follows
\begin{equation}
	I(X,Y)=h(Y)-h(Y|X)=h(Y)-h(N),
\end{equation}
where $h(\cdot)$ denotes the entropy operator.
When the power of noise $N$ is known, $h(N)$ is easily derived as
\begin{equation}
	h(N)=\frac{d}{2}\ln2\pi {\rm e} \sigma^2,
\end{equation}
where the notations are the same as those in Eq. (\ref{eq.I_estimation}).
Hence, we only need to estimate $h(Y)$. For the discrete case, a closed form of the distribution of $Y$ is given by
\begin{equation}
	p_Y(\mathbf{y}) = \frac{1}{M}\sum_{i=1}^M \frac{1}{(2\pi\sigma^2)^{\frac{d}{2}}}\exp(-\frac{||\mathbf{y}-\mu_i||^2}{2\sigma^{2d}}),
\end{equation}
in which $\mu_i$ ($i=1,\cdots,M$) denote the constellation points.
Since $h(Y)=\mathbb{E}_{\mathbf{y} \sim p_Y(\mathbf{y})}[-\ln p_Y(\mathbf{y})]$, it is not difficult to estimate $h(Y)$ through Monte Carlo methods.
\begin{figure}[htbp]
    \centering
    \includegraphics[scale=0.5]{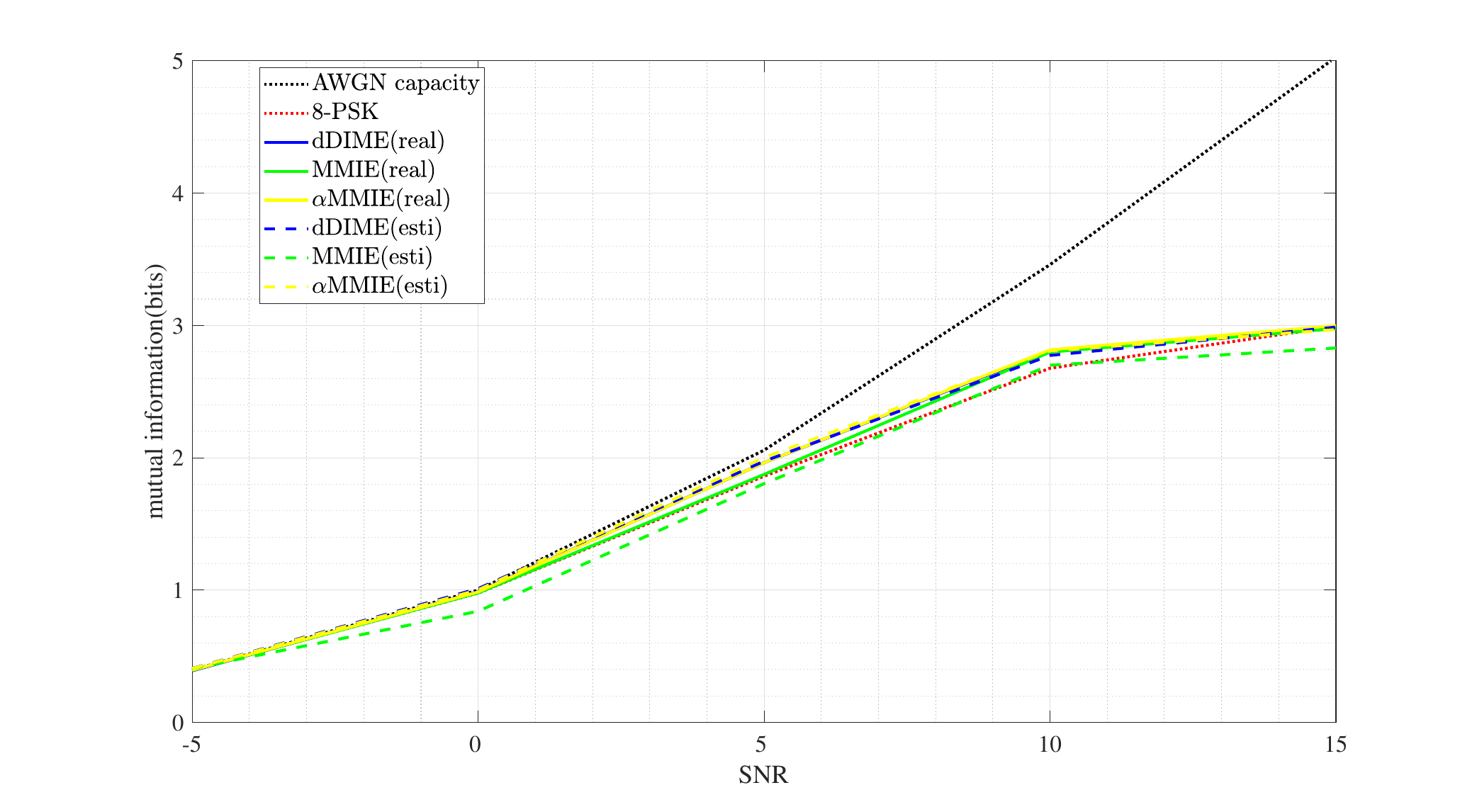}\\
    \caption{Performance of the cooperative framework for a discrete channel input ($I(X,Y)$ versus SNR).
    }
    \label{fig_4}
\end{figure}
\begin{figure}[htbp]
	\centering
	\includegraphics[scale=0.9]{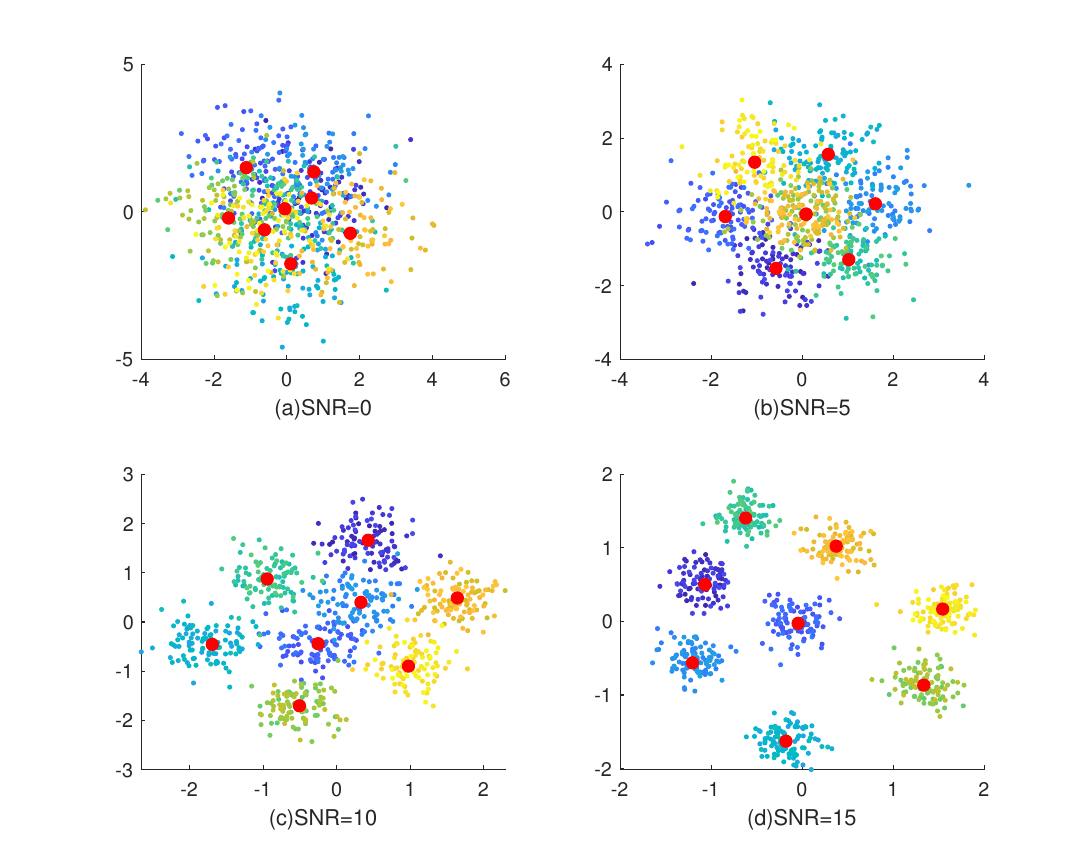}\\
	\caption{Codebook designed by $\alpha$-MMIE under different SNR values for an AWGN channel.}
	\label{fig_5}
\end{figure}

Fig. \ref{fig_4} illustrates the mutual information with the given channel input designed by the generators, and shows the codebook of the channel input.
In this figure, the supportive estimator is viewed as the ground truth channel capacity.
The codebooks given by dDIME and $\alpha$-MMIE perform almost the same, and better than the MMIE. Nevertheless, all the three estimators produce codebooks that surpass that from $8$-PSK. Besides, MMIE tends to give a lower estimation for the channel capacity.

Fig. \ref{fig_5} shows the codebook designed by the $\alpha$-MMIE under different SNR values. The red points denote the channel input or the codebook designed by our method, while the points with other colors represent the samples of channel output. It is readily seen that the codebook has different patterns under different SNRs. Besides, when $\text{SNR}=5\text{dB}$, the two red points coincide in the center of the codebook.

\section{Conclusion}
In this paper, we have proposed a MIM-based mutual information estimator, which is referred to as MMIE. First of all, we have designed the non-parametric form of MMIE and showed its relationship with the  MIM-based GAN. We have also analyzed its convergence and relationship with the R\'{e}nyi divergence. It is then concluded that the performance degradation for the GAN-based estimators may result from the deviation of the discriminator output related to the initial settings. To overcome this issue, the parametric MMIE, named $\alpha-\text{MMIE}$, was presented and its convergence was also analyzed in theory. A guideline for selecting the parameter $\alpha$ was given based on the estimation of SNR. Furthermore, a cooperative framework based on MMIE was investigated to learn the channel capacity and design the corresponding codebook. Finally, we showed that MMIE provides better accuracy and stability for estimating the mutual information and channel capacity through numerical simulations.

\section*{Acknowledgment}
The authors would like to appreciate the support of the National Key R \& D Program of China No. 2021YFA1000504. The authors thank the members of Wistlab of Tsinghua University for their good suggestions and discussions.

\appendices
\section{Proof of Theorems}
\subsection{Theorems \ref{thm.I_MMIE} and \ref{thm.I_alpha_MMIE}}\label{app.MMIE}
Since Theorem \ref{thm.I_MMIE} is a special case of Theorem \ref{thm.I_alpha_MMIE} with parameter $\alpha=1$, we only need to prove Theorem \ref{thm.I_alpha_MMIE}.
Consider the function $f(u)=c \exp(\alpha-u)+d\exp(u-\alpha), c>0, d>0$. Then, the first and second derivatives with respect to $u$ are given by
\begin{subequations}
\begin{align}
	& f'(u)=-c\exp(\alpha-u)+d\exp(u-\alpha), \\
	& f''(u)=c\exp(\alpha-u)+d\exp(u-\alpha).
\end{align}
\end{subequations}
Since $f''(u)>0$ for any $u$ ($|u|< +\infty$), $f$ is a lower convex function and achieves its minimum if and only if $f'(u)=0$.
Then, we have
\begin{equation}
	u=\frac{1}{2}\ln\frac{c}{d}+\alpha.
\end{equation}

It is not difficult to see that Eq. (\ref{eq.D_*_alpha_MMIE}) holds.
From the definition of the mutual information \cite{cover1999elements}, we have
\begin{equation}
	\begin{split}
		I(X,Y)
        &=\mathbb{E}_{\mathbf{x},\mathbf{y}\sim p_{XY}(\mathbf{x},\mathbf{y})}\left[\ln\frac{p_{XY}(\mathbf{x},\mathbf{y})}{p_X(\mathbf{x})p_Y(\mathbf{y})}\right]  \\
		&=2\mathbb{E}_{\mathbf{x},\mathbf{y}\sim p_{XY}(\mathbf{x},\mathbf{y})}[D(\mathbf{x},\mathbf{y})]-2\alpha.
	\end{split}
\end{equation}
Therefore, the proof is already complete.

\subsection{Proposition \ref{prop.MMIE_Renyi}}\label{app.MMIE_Renyi}
By substituting the optimal discriminator given by Eq. (\ref{eq.MMIE_D^*}) into Eq. (\ref{eq.J_MIM_D}), we have
\begin{equation}
	\begin{split}
		& J_{\text{MIM}}(D^*_{\text{\rm MMIE}}) \\
        =& \mathbb{E}_{(\mathbf{x},\mathbf{y})\sim p_{XY}(\mathbf{x},\mathbf{y})}\left[\exp(\frac{1}{2}-\frac{1}{2}\ln\frac{p_{XY}(\mathbf{x},\mathbf{y})}{p_X(\mathbf{x})p_Y(\mathbf{y})})\right] +\mathbb{E}_{(\mathbf{x},\mathbf{y})\sim p_X(\mathbf{x})p_Y(\mathbf{y})}\left[\exp(\frac{1}{2}+\frac{1}{2}\ln\frac{p_{XY}(\mathbf{x},\mathbf{y})}{p_X(\mathbf{x})p_Y(\mathbf{y})})\right] \\
		=&\sqrt{{\rm e}}\mathbb{E}_{(\mathbf{x},\mathbf{y})\sim p_X(\mathbf{x})p_Y(\mathbf{y})}\left[\sqrt{\frac{p_{XY}(\mathbf{x},\mathbf{y})}{p_X(\mathbf{x})p_Y(\mathbf{y})}}\right] +\sqrt{{\rm e}}\mathbb{E}_{(\mathbf{x},\mathbf{y})\sim
p_{XY}(\mathbf{x},\mathbf{y})}\left[\sqrt{\frac{p_X(\mathbf{x})p_Y(\mathbf{y})}{p_{XY}(\mathbf{x},\mathbf{y})}}\right] \\
		=&\sqrt{{\rm e}}\int_{\mathcal{X}}\int_{\mathcal{Y}}\sqrt{\frac{p_{XY}(\mathbf{x},\mathbf{y})}{p_X(\mathbf{x})p_Y(\mathbf{y})}}p_X(\mathbf{x})p_Y(\mathbf{y})
{\rm d}\mathbf{x}{\rm d}\mathbf{y} +\sqrt{{\rm e}}\int_{\mathcal{X}}\int_{\mathcal{Y}}\sqrt{\frac{p_X(\mathbf{x})p_Y(\mathbf{y})}{p_{XY}(\mathbf{x},\mathbf{y})}}p_{XY}(\mathbf{x},\mathbf{y})
{\rm d}\mathbf{x}{\rm d}\mathbf{y} \\
		=&2\sqrt{{\rm e}}\int_{\mathcal{X}}\int_{\mathcal{Y}}\sqrt{p_{XY}(\mathbf{x},\mathbf{y})p_X(\mathbf{x})p_Y(\mathbf{y})}
{\rm d}\mathbf{x}{\rm d}\mathbf{y}.
	\end{split}
\end{equation}

Consider
\begin{equation}
	\begin{split}
		\mathcal{R}_{\frac{1}{2}}(P\|Q)&=\frac{1}{\frac{1}{2}-1}\ln\mathbb{E}_{\mathbf{x}\sim P}\left[\left(\frac{P(\mathbf{x})}{Q(\mathbf{x})}\right)^{\frac{1}{2}-1}\right] \\
		&=-2\ln\int_{\mathcal{X}}\sqrt{\frac{Q(\mathbf{x})}{P(\mathbf{x})}}P(\mathbf{x}){\rm d}{\mathbf{x}} \\
		&=-2\ln\int_{\mathcal{X}}\sqrt{P(\mathbf{x})Q(\mathbf{x})}{\rm d}{\mathbf{x}},
	\end{split}
\end{equation}
then we have
\begin{equation}
	\int_{\mathcal{X}}\sqrt{P(\mathbf{x})Q(\mathbf{x})}{\rm d}\mathbf{x}=\exp(-\frac{1}{2}\mathcal{R}_{\frac{1}{2}}(P\|Q)).
\end{equation}

Hence, it is given by
\begin{equation}
	\begin{split}
		J_{\text{MIM}}(D^*_{\text{\rm MMIE}})&=2\sqrt{{\rm e}}\exp(-\frac{1}{2}\mathcal{R}_{\frac{1}{2}}(p_{XY}\|p_Xp_Y)) \\
		&=2\exp\left(\frac{1}{2}(1-\mathcal{R}_{\frac{1}{2}}(p_{XY}\|p_Xp_Y))\right) \\
        & \le J_{\text{MIM}}(D),
	\end{split}	
\end{equation}
from which, it is readily seen that (\ref{eq.R_JMIM}) holds.
Therefore, this proof is complete.

\section{Experiment Details}
In this section, we describe the details of the experiments.
\begin{table}[h]
	\centering
	\caption{\bf Discriminator Network Architecture and Training Parameters}\label{table.1}
	\begin{tabular}{|p{3.8cm}|p{1.9cm}|p{1.5cm}|}
		\hline
        \rowcolor{mygray}
        \textbf{Layer} & \textbf{output dimension} & \textbf{activation function} \\
		\hline
		MINE \& MMIE \& $\alpha$MMIE& & \\
		input [$\mathbf{x}$, $\mathbf{y}$] &2d & \\
		Fully connected &100 &ReLU \\
		Dropout&0.3 & \\
		Fully connected &100 &ReLU \\
		Fully connected &1 &Linear \\
		\hline
		iDIME& & \\
		input [$\mathbf{x}$, $\mathbf{y}$] &2d & \\
		Fully connected &100 &ReLU \\
		Dropout&0.3 & \\
		Fully connected &100 &ReLU \\
		Fully connected &1 &Sigmoid \\
		\hline
		dDIME& & \\
		input [$\mathbf{x}$, $\mathbf{y}$] &2d & \\
		Fully connected &100 &ReLU \\
		Dropout&0.3 & \\
		Fully connected &100 &ReLU \\
		Fully connected &1 &Softplus \\
		\hline
		Batch size &\multicolumn{2}{c|}{512} \\
		Number of training iteration &\multicolumn{2}{c|}{400} \\
		Learning rate & \multicolumn{2}{c|}{0.002} \\
		Optimizer & \multicolumn{2}{c|}{Adam ($\beta_1$=0.5, $\beta_2$=0.999)} \\
		\hline
	\end{tabular}
\end{table}
\begin{table}
	\centering
	\caption{\bf Generator Network Architecture and Training Parameters}\label{table.2}
	\begin{tabular}{|p{3.8cm}|p{1.9cm}|p{1.5cm}|}
		\hline
        \rowcolor{mygray}
        \textbf{Layer} & \textbf{output dimension} & \textbf{activation function} \\
		\hline
		Generator& & \\
		input z & 30 (continuous)/ 3 (discrete) & \\
		Fully connected &100 &ReLU \\
		Fully connected &100 &ReLU \\
		Fully connected &100 &ReLU \\
		Fully connected &1 &Linear \\
		Batch Normalization & & \\
		\hline
		Batch size &\multicolumn{2}{c|}{512} \\
		Number of training iteration &\multicolumn{2}{c|}{10000} \\
		Learning rate & \multicolumn{2}{c|}{0.0001} \\
		Optimizer & \multicolumn{2}{c|}{Adam ($\beta_1$=0.5, $\beta_2$=0.999)} \\
		\hline
	\end{tabular}
\end{table}

\subsection{Network Architecture}
The architecture of the neural networks mainly refers to \cite{r9}. A two-hidden-layers multilayer perceptron (MLP) neural network is adopted as the backbone for all the discriminators or estimators with 100 hidden units in each hidden layer and ReLU as the activation function. The estimators differ only in the activation function for the additional third hidden layer and the last layer. Specifically, on one hand, we choose a linear function for MINE, MMIE and $\alpha-\text{MMIE}$, while the Sigmoid function for iDIME and softplus function for dDIME. On the other hand, the generator in the cooperative framework has a three-hidden-layers multilayer perceptron (MLP) neural network with 100 hidden units in each layer and ReLU as the activation function. Then, a linear layer follows the backbone and allows arbitrary input distribution. The final layer for the generator is a batch normalization layer to restrict the power of the channel input within one. More details are shown in TABLE \ref{table.1} and \ref{table.2}.

\subsection{Training and Testing Details}
	In the first part of the experiment, we compare the performance of different mutual information estimators when the dimension of the channel input is set as $d=2$ and $d=10$. For each SNR, the estimator is trained for $5000$ iterations with a batch size of $512$. The optimizer is chosen as Adam with parameters $\beta_1=0.5$ and $\beta_2=0.999$, where the learning rate is set as $0.002$. To avoid a training bias, we train $100$ estimators independently for each approach. In the testing process, each estimator is based on $1000$ test sets with a batch size of $512$. Hence, for each estimator under a given SNR, we obtain a total of $100000$ estimated values of the mutual information.
	
	In the second part of the experiment, we apply the cooperative framework to estimate the channel capacity, as well as design the codebook. For the discrete coding case, $\mathbf{z}$ is sampled from a $3$-d Bernoulli distribution with the parameter $p=0.5$.

\ifCLASSOPTIONcaptionsoff
  \newpage
\fi



%
\bibliographystyle{IEEEtran}
\bibliography{reference}

\begin{thebibliography}{10}
\providecommand{\url}[1]{#1}
\csname url@samestyle\endcsname
\providecommand{\newblock}{\relax}
\providecommand{\bibinfo}[2]{#2}
\providecommand{\BIBentrySTDinterwordspacing}{\spaceskip=0pt\relax}
\providecommand{\BIBentryALTinterwordstretchfactor}{4}
\providecommand{\BIBentryALTinterwordspacing}{\spaceskip=\fontdimen2\font plus
\BIBentryALTinterwordstretchfactor\fontdimen3\font minus
  \fontdimen4\font\relax}
\providecommand{\BIBforeignlanguage}[2]{{%
\expandafter\ifx\csname l@#1\endcsname\relax
\typeout{** WARNING: IEEEtran.bst: No hyphenation pattern has been}%
\typeout{** loaded for the language `#1'. Using the pattern for}%
\typeout{** the default language instead.}%
\else
\language=\csname l@#1\endcsname
\fi
#2}}
\providecommand{\BIBdecl}{\relax}
\BIBdecl

\bibitem{r1}
N.~Harvey, R.~Kleinberg, and A.~R. Lehman, ``On the capacity of information
  networks,'' \emph{IEEE Trans. Inf. Theory}, vol.~52, no.~6, pp. 2345--2364,
  2006.

\bibitem{r2}
S.~Verdu and T.~S. Han, ``A general formula for channel capacity,'' \emph{IEEE
  Trans. Inf. Theory}, vol.~40, no.~4, pp. 1147--1157, 2008.

\bibitem{r3}
R.~H. Etkin, D.~Tse, and H.~Wang, ``Gaussian interference channel capacity to
  within one bit,'' \emph{IEEE Trans. Inf. Theory}, vol.~54, no.~12, pp.
  5534--5562, 2008.

\bibitem{r4}
Durisi and Giuseppe, ``On the capacity of the block-memoryless phase-noise
  channel,'' \emph{IEEE Commun. Lett.}, vol.~16, no.~8, pp. 1157--1160, 2012.

\bibitem{r5}
R.~Jiang, Z.~Wang, Q.~Wang, and L.~Dai, ``A tight upper bound on channel
  capacity for visible light communications,'' \emph{IEEE Commun. Lett.},
  vol.~20, no.~1, pp. 97--100, 2016.

\bibitem{r6}
S.~M. Shah and V.~Sharma, ``Achieving {Shannon} capacity in a wiretap channel
  via previous messages,'' \emph{arXiv preprint arXiv:1404.5701}, 2014.

\bibitem{r7}
C.~Zhang and R.~C. Qiu, ``Massive {MIMO} testbed-implementation and initial
  results in system model validation,'' \emph{arXiv preprint arXiv:1501.00035},
  2014.

\bibitem{r8}
Shannon and E.~C., ``A mathematical theory of communication,'' \emph{Bell
  Systems Technical Journal}, vol.~27, no.~4, pp. 623--656, 1948.

\bibitem{r9}
N.~A. Letizia and A.~M. Tonello, ``Discriminative mutual information estimators
  for channel capacity learning,'' \emph{arXiv preprint arXiv:2107.03084},
  2021.

\bibitem{nguyen2010estimating}
{Nguyen, XuanLong and Wainwright, Martin J and Jordan, Michael I},
  ``{Estimating divergence functionals and the likelihood ratio by convex risk
  minimization},'' \emph{{IEEE Trans. Inf. Theory}}, vol.~{56}, no.~{11}, pp.
  {5847--5861}, {2010}.

\bibitem{wunder2021reverse}
{Wunder, Gerhard and Gro{\ss}, Benedikt and Fritschek, Rick and Schaefer,
  Rafael F}, ``{A reverse Jensen inequality result with application to mutual
  information estimation},'' in \emph{{2021 IEEE Information Theory Workshop
  (ITW)}}.\hskip 1em plus 0.5em minus 0.4em\relax {IEEE}, {2021}, pp. {1--6}.

\bibitem{r10}
M.~e.~a. Frank, ``An algorithm for quadratic programming,'' \emph{Naval
  research logistics quarterly}, vol.~3, no. 1-2, pp. 95--110, 1956.

\bibitem{r11}
J.~Cao, T.~Zhao, and Z.~Chen, ``{Frank-Wolfe} method used for solving channel
  capacity,'' \emph{Journal of Yangtze University (Natural Science Edition) Sci
  \& Eng V}, p.~3, 2008.

\bibitem{r12}
R.~Blahut, ``Computation of channel capacity and rate-distortion functions,''
  \emph{IEEE Trans. Inf. Theory}, vol.~18, no.~4, pp. 460--473, 1972.

\bibitem{r13}
Y.~Yu, ``Squeezing the {Arimoto-Blahut} algorithm for faster convergence,''
  \emph{IEEE Trans. Inf. Theory}, vol.~56, no.~7, pp. 3149--3157, 2010.

\bibitem{r14}
D.~M. Arnold, H.~A. Loeliger, P.~O. Vontobel, A.~Kavcic, and W.~Zeng,
  ``Simulation-based computation of information rates for channels with
  memory,'' \emph{IEEE Trans. Inf. Theory}, vol.~52, no.~8, pp. p.3498--3508,
  2006.

\bibitem{r25}
H.~D. Pfister, J.~B. Soriaga, and P.~H. Siegel, ``On the achievable information
  rates of finite state {ISI} channels,'' in \emph{IEEE Global Communications
  Conference (GLOBECOM)}, vol.~5, 2001, pp. 2992--2996.

\bibitem{r15}
M.~I. Belghazi, A.~Baratin, S.~Rajeswar, S.~Ozair, Y.~Bengio, A.~Courville, and
  R.~D. Hjelm, ``{MINE}: Mutual information neural estimation,'' \emph{arXiv
  preprint arXiv:1801.04062}, 2018.

\bibitem{r21}
J.~Song and S.~Ermon, ``Understanding the limitations of variational mutual
  information estimators,'' \emph{arXiv preprint arXiv:1910.06222}, 2019.

\bibitem{r28}
A.~v.~d. Oord, Y.~Li, and O.~Vinyals, ``Representation learning with
  contrastive predictive coding,'' \emph{arXiv preprint arXiv:1807.03748},
  2018.

\bibitem{r26}
R.~Fritschek, R.~F. Schaefer, and G.~Wunder, ``Deep learning for channel coding
  via neural mutual information estimation,'' in \emph{IEEE International
  Workshop on Signal Processing Advances in Wireless Communications (SPAWC)},
  2019, pp. 1--5.

\bibitem{r27}
N.~A. Letizia and A.~M. Tonello, ``Capacity-approaching autoencoders for
  communications,'' \emph{arXiv preprint arXiv:2009.05273}, 2020.

\bibitem{r29}
F.~Mirkarimi and N.~Farsad, ``Neural computation of capacity region of
  memoryless multiple access channels,'' in \emph{IEEE International Symposium
  on Information Theory (ISIT)}, 2021, pp. 2066--2071.

\bibitem{r30}
{Aharoni, Ziv and Tsur, Dor and Goldfeld, Ziv and Permuter, Haim H.},
  ``{Capacity of continuous channels with memory via directed information
  neural estimator},'' in \emph{{IEEE International Symposium on Information
  Theory (ISIT)}}, {2020}, pp. {2014--2019}.

\bibitem{r17}
A.~Hacine-Gharbi and P.~Ravier, ``A binning formula of bi-histogram for joint
  entropy estimation using mean square error minimization,'' \emph{Pattern
  Recognition Lett.}, vol. 101, pp. 21--28, 2018.

\bibitem{r18}
N.~Bi, J.~Tan, J.-H. Lai, and C.~Y. Suen, ``High-dimensional supervised feature
  selection via optimized kernel mutual information,'' \emph{Expert Syst.
  Appl.}, vol. 108, pp. 81--95, 2018.

\bibitem{r19}
A.~Gretton, R.~Herbrich, and A.~J. Smola, ``The kernel mutual information,'' in
  \emph{IEEE International Conference on Acoustics, Speech, and Signal
  Processing (ICASSP)}, vol.~4, 2003, pp. 880--883.

\bibitem{r20}
M.~Hulle, ``Edgeworth approximation of multivariate differential entropy,''
  \emph{Neural Comput.}, vol.~17, no.~9, pp. 1903--1910, 2005.

\bibitem{r31}
{Mirkarimi, Farhad and Rini, Stefano and Farsad, Nariman}, ``{Neural capacity
  estimators: How reliable are they?}'' in \emph{{IEEE International Conference
  on Communications (ICC)}}, {2022}, pp. {3868--3873}.

\bibitem{r16}
R.~She and P.~Fan, ``{MIM-Based GAN}: information metric to amplify small
  probability events importance in generative adversarial networks,''
  \emph{arXiv preprint arXiv:2003.11285}, 2020.

\bibitem{r22}
D.~e.~a. Agakov, ``The {IM} algorithm: A variational approach to information
  maximization,'' \emph{Advances in Neural Information Processing Systems
  (NeurIPS)}, vol.~16, no. 320, p. 201, 2004.

\bibitem{r32}
{Chung Chan and Ali Al-Bashabsheh and Hingpang Huang and Michael Lim and Da Sun
  Handason Tam and Chao Zhao}, ``{{Neural Entropic Estimation}: A faster path
  to mutual information estimation},'' \emph{{ArXiv}}, vol. {abs/1905.12957},
  {2019}.

\bibitem{r33}
{H{\"a}ger, Christian and Agrell, Erik}, ``{Data-driven estimation of capacity
  upper bounds},'' \emph{{IEEE Commun. Lett.}}, vol.~{26}, no.~{12}, pp.
  {2939--2943}, {2022}.

\bibitem{r23}
I.~Goodfellow, J.~Pouget-Abadie, M.~Mirza, B.~Xu, D.~Warde-Farley, S.~Ozair,
  A.~Courville, and Y.~Bengio, ``Generative adversarial nets,'' \emph{Advances
  in Neural Information Processing Systems (NeurIPS)}, vol.~27, pp. 2672--2680,
  2014.

\bibitem{r24}
P.~Fan, Y.~Dong, J.~Lu, and S.~Liu, ``Message importance measure and its
  application to minority subset detection in big data,'' in \emph{IEEE
  GLOBECOM Workshops (GC Wkshps)}, 2016, pp. 1--5.

\bibitem{r34}
{Liu, Shanyun and She, Rui and Fan, Pingyi}, ``{Differential message importance
  measure: A new approach to the required sampling number in big data structure
  characterization},'' \emph{{IEEE Access}}, vol.~{6}, pp. {42\,851--42\,867},
  {2018}.

\bibitem{r35}
{Liu, Shanyun and She, Rui and Fan, Pingyi and Letaief, Khaled B.},
  ``{Non-parametric message importance measure: Storage code design and
  transmission planning for big data},'' \emph{{IEEE Trans. Commun.}},
  vol.~{66}, no.~{11}, pp. {5181--5196}, {2018}.

\bibitem{r38}
R.~She and P.~Fan, ``From mim-based gan to anomaly detection: Event probability
  influence on generative adversarial networks,'' \emph{IEEE Internet of Things
  Journal}, vol.~9, no.~19, pp. 18\,589--18\,606, 2022.

\bibitem{r39}
R.~She, P.~Fan, X.-Y. Liu, and X.~Wang, ``Interpretable generative adversarial
  networks with exponential function,'' \emph{IEEE Transactions on Signal
  Processing}, vol.~69, pp. 3854--3867, 2021.

\bibitem{r36}
{Ruder, Sebastian}, ``{An overview of gradient descent optimization
  algorithms},'' \emph{{arXiv preprint arXiv:1609.04747}}, {2016}.

\bibitem{cover1999elements}
T.~M. Cover, \emph{Elements of information theory}.\hskip 1em plus 0.5em minus
  0.4em\relax New York, U.S.: John Wiley \& Sons, 1999.

\bibitem{r37}
{Kraskov, Alexander and St{\"o}gbauer, Harald and Grassberger, Peter},
  ``{Estimating mutual information},'' \emph{{Physical review E}}, vol.~{69},
  no.~{6}, p. {066138}, {2004}.

\end{thebibliography}
%







\end{document}